\documentclass[a4paper,USenglish,cleveref,autoref,thm-restate]{lipics-v2021}
\pdfoutput=1 
\hideLIPIcs  
\usepackage{algorithm}
\usepackage{algpseudocode}
\usepackage{xspace}
\usepackage{xparse}
\usepackage{float}
\usepackage{cite}
\usepackage{hyperref}
\algnewcommand{\algorithmicgoto}{\textbf{go to}}%
\algnewcommand{\Goto}{\algorithmicgoto\xspace}%
\algnewcommand{\Label}{\State\unskip}
\algnewcommand{\LineComment}[1]{{\ttfamily \footnotesize \(\triangleright\) #1}}
\algrenewcommand\algorithmiccomment[1]{{\ttfamily \footnotesize \hfill \(\triangleright\) #1}}
\newdimen{\algindent}
\algnewcommand{\LeftComment}[2]{\hspace{#1\algindent} {\ttfamily \footnotesize \(\triangleright\) #2}}
\newcommand{\eqtag}[1]{\stepcounter{equation}\tag{\theequation: #1}} 
\NewDocumentEnvironment{retheorem}{mo}
    {   \begingroup
        \def\thetheorem{#1}
        \IfNoValueTF{#2} 
            {\begin{theorem}} 
            {\begin{theorem}[#2]}
    }
    {   \end{theorem}
        \addtocounter{theorem}{-1}
        \endgroup
    }

\NewDocumentEnvironment{relemma}{mo}
    {   \begingroup
        \def\thelemma{#1}
        \IfNoValueTF{#2} 
            {\begin{lemma}} 
            {\begin{lemma}[#2]}
    }
    {   \end{lemma}
        \addtocounter{lemma}{-1}
        \endgroup
    }

\newcommand{\RE}{\mathbb{R}} 
\newcommand{\ZZ}{\mathbb{Z}} 
\newcommand{\eps}{\varepsilon} 
\newcommand{\etal}{\textit{et al.}}
\newcommand{\SP}{\kern+1pt} 

\DeclareMathOperator{\vol}{vol}

\DeclareMathOperator{\radius}{radius}

\DeclareMathOperator{\dist}{\mathcal D (\mathcal P, \mathcal H )}

\newcommand{\wt}[1]{\widetilde{#1}}
\newcommand{\mc}{\mathcal}

\newcommand{\mf}{\mathscr}
\newcommand{\vctr}{\mathbf} 
\newcommand{\spr}[1]{^{(#1)}}
\newcommand{\itr}[1]{^{\langle#1\rangle}}
\newcommand{\adi}[1]{}
\newcommand{\dave}[1]{}
\newcommand{\arxivonly}[1]{}

\title{Evolving Distributions under Local Motion} 

\titlerunning{Evolving Distributions under Local Motion} 

\author{Aditya Acharya}{Department of Computer Science, University of Maryland, College Park, USA}{adach@umd.edu}{https://orcid.org/0000-0002-0359-1913}{}

\author{David M. Mount}{Department of Computer Science and Institute for Advanced Computer Studies, University of Maryland, College Park, USA}{mount@umd.edu}{https://orcid.org/0000-0002-3290-8932}{}

\authorrunning{A.\ Acharya and D.\ M.\ Mount}
\Copyright{Aditya Acharya and David M. Mount}


\nolinenumbers 

\begin{document}

\maketitle

\begin{abstract}
Geometric data sets arising in modern applications are often very large and change dynamically over time. A popular framework for dealing with such data sets is the evolving data framework, where a discrete structure continuously varies over time due to the unseen actions of an evolver, which makes small changes to the data. An algorithm probes the current state through an oracle, and the objective is to maintain a hypothesis of the data set's current state that is close to its actual state at all times. In this paper, we apply this framework to maintaining a set of $n$ point objects in motion in $d$-dimensional Euclidean space. To model the uncertainty in the object locations, both the ground truth and hypothesis are based on spatial probability distributions, and the distance between them is measured by the Kullback-Leibler divergence (relative entropy). We introduce a simple and intuitive motion model where, with each time step, the distance that any object can move is a fraction of the distance to its nearest neighbor. We present an algorithm that, in steady state, guarantees a distance of $O(n)$ between the true and hypothesized placements. We also show that for any algorithm in this model, there is an evolver that can generate a distance of $\Omega(n)$, implying that our algorithm is asymptotically optimal. 
\end{abstract}

\section{Introduction} \label{sec:intro}

Many modern computational applications are characterized by two qualities: data sets are massive and they vary over time. A fundamental question is how to maintain some combinatorial structure that is a function of such a data set. The combination of size and dynamics makes maintaining such a structure challenging. Given the large data sizes, single-shot algorithms may be too slow, and common dynamic algorithms, which support explicit requests for insertions and deletions, may not be applicable because structural changes are unseen by the algorithm.

Anagnostopoulos \etal introduced a model for handling such data sets, called the \emph{evolving data framework}~\cite{AKM11}. In this framework, the structure continuously varies over time through the unseen actions of an \emph{evolver}, which makes small, stochastic changes to the data set. The algorithm can probe the current state locally through an \emph{oracle}. With the aid of this oracle, the algorithm attempts to maintain a \emph{hypothesis} of the structure's current state that is as close as possible to its actual state, at all times. The similarity between the hypothesis and current state is measured through some \emph{distance function}. The algorithm's objective is to achieve, in the steady-state, a small distance between the hypothesis and the actual state. This framework has been applied to a variety of problems~\cite{AcM22, BDE18a, KLM16, HLS17, ZZL16}.

Consider sorting for example. The data consists of a set of objects over some total order. The evolver continuously selects a random pair of adjacent objects and swaps them. The oracle is given two objects and returns their relative order. The algorithm's objective is to maintain an order that is as close to the current state as possible, where the distance is measured in terms of the Kendall tau distance, that is, the number of pairwise order inversions \cite{Ken38}. It has been shown that a Kendall tau distance of $O(n)$ is achievable, and this is optimal \cite{BDE18a, BDE18b}. It is important to mention that quadratic time algorithms—not the sorting algorithms that are optimal in static scenarios—achieve this distance, which further underscores why the dynamic framework can be contrary to intuition.
The sorting problem has been generalized to tracking labels on a tree in \cite{AcM22}, which lays the foundations for a geometric framework for evolving data.

This paper focuses on the question of how to maintain a set of points whose positions evolve continuously in real $d$-dimensional space, $\RE^d$. In motion-tracking applications, object motion is recorded through various technologies, including GPS-enabled mobile devices \cite{TFK10}, RFID tags \cite{PZT23}, and camera-based sensing~\cite{YLW15}. Examples include the movement and migration of animals on land and in oceans, traffic and transport, defense and surveillance, and analysis of human behavior (see, e.g., \cite{GLW12, Lau14}).

Two essential limiting features of tracking technologies are the cost of data acquisition and imprecision in the data. The cost of data acquisition is dominated by the cost of communicating the positions of the objects being tracked, which can be measured in terms of the number of bits needed to transmit object locations to the processing algorithm. Imprecision arises due to time delays in gathering and reporting positions combined with the fact that objects are in continuous motion. Hence, the exact location of any object can never be known with certainty. This implies that, at best, the algorithm can maintain only an approximation to the current state.

In order to have any hope of bounding the degree of uncertainty, it is necessary to impose restrictions on object motion. In traditional applications of the evolving data framework, the evolver acts randomly. This is not a reasonable assumption in practice, where moving objects are subject to physical laws or may have a sense of agency \cite{HBJ05, KLS18, WLS18}. Much work has focused on realistic models of motion, but these can be difficult to analyze theoretically. A simple, clean model that has been analyzed theoretically assumes that objects have a maximum speed limit, and as time passes an object can be inferred to lie within a ball whose radius grows linearly as a function of the elapsed time since the object's last reported position \cite{BEK19, EvK24, Kah91a}. However, this makes the strong assumptions that an object's motion can be reasonably modeled in terms of a fixed maximum speed limit, independent of its environment. 

In multidimensional space, there is no intrinsic total order, and it is less clear what it means to accurately track imprecise moving objects. Given the inherent uncertainty, we model both the truth, that is, each object's position, and our hypothesis as \emph{spatial probability distributions}. The distance between the actual state of the system and our hypothesis is naturally defined as the \emph{relative entropy} (Kullback-Leibler divergence) between these two distributions. The \emph{Kullback-Leibler} (KL) \emph{divergence} is a fundamental measure in statistics representing the distance between two probability distributions \cite{KL51}. It has numerous applications in statistical inference, coding theory, machine learning among others \cite{CT2012}. In our case, it serves two intuitive purposes. First, it measures how different the truth is from our hypothesis. Second, given its application in coding theory relating to coding penalty \cite{CT2012}, it quantifies the additional information required to encode the actual location of the objects utilizing our hypothesis, in contrast to an entity possessing access to the true distributions.

For the evolution of our data, we adopt a locality-sensitive stochastic motion model. We assume that with each step, an object's motion is constrained by its immediate environment, which we call the \emph{local-motion model}. In this model, the distance that any object can move in a single time step is a fixed fraction of the distance to its nearest neighbor in the set. The support of this object's probability distribution is a region of size proportional to the nearest-neighbor distance. This model has the advantage that it does not impose arbitrary speed limits on the objects, it is invariant under transformations that preserve relative distances (e.g., translation, rotation, and uniform scaling), and it satisfies the observed phenomenon that objects in dense environments have less personal space \cite{Jain1987,HW17}, and exhibit slower movement than those in sparse environments \cite{SSK05,BS23}.

To control the communication complexity in determining object locations, we do not require that the oracle returns the exact object positions. Instead, we present the oracle a Euclidean ball in the form of a center point and radius and the index $i$ of an object. The oracle returns a pair of bits indicating whether the $i$-th object lies within this ball, and (if so) whether there is any other object of the set within this ball, giving us a rough estimate of the distance to the object's nearest neighbor. Since every object's location is subject to uncertainty, the same query on the oracle might result in different outcomes at different times. Our algorithm is robust to such variations. Note however, since we are only interested in containment queries, the exact center and the radius are not critical to our algorithm, further driving down the communication complexity.

In our framework, the evolver and the algorithm operate asynchronously and in parallel. With each step, the evolver selects an arbitrary object of the set and moves it in accordance with the local-motion model. (This need not be random, and may even be adversarial.) This information is hidden from our algorithm. The algorithm selects an object and invokes the oracle on this object. Based on the oracle's response, the algorithm updates the current hypothesized distribution for this point. Thus, the evolver and algorithm are involved in a pursuit game, with the evolver incrementally changing object distributions (possibly in an adversarial manner) and the algorithm updating its hypothesized distributions. The algorithm's goal is to minimize the relative entropy between these distributions, at all times. Our computational model and a formal statement of our results are presented in Section~\ref{sec:problem}.

\section{Problem Formulation and Results} \label{sec:problem}

\paragraph*{Objects} 

The objects that populate our system can be thought of as imprecise points \cite{BLM11, GSS89, LoK10} that are drawn independently from probability distributions that depend on the proximity to other objects. More formally, let $Q =\{\vctr q_1, \ldots, \vctr q_n\} \subset \RE^d$, be a point set of size $n$ in $d$-dimensional Euclidean space, and let $[n]$ denote the index set $\{1, \ldots, n\}$. For each $\vctr q_i \in Q$, let $N_i$ denote the distance to its nearest neighbor in $Q \setminus \{\vctr q_i\}$. Given $\vctr x \in \RE^d$ and nonnegative real $r$, let $\mc B(\vctr x, r)$ denote the closed Euclidean ball of radius $r$ centered at $\vctr x$, and let $\mathrm U(\mc B(\vctr x, r))$ denote the uniform probability distribution over this ball. Given a real $\beta$, where $0 < \beta < 1$, define the \emph{$\beta$-local feature region} of $\vctr q_i$ to be $\mc B(\vctr q_i,\beta N_i)$.
Let $P_i = P_i(\beta)$ denote the uniform probability distribution: $\mathrm U \left(\mc B(\vctr q_i,\beta N_i)\right)$, and let $\mc P = \{P_1, \ldots, P_n\}$. We refer to $\mc P$ as the \emph{truth} or \emph{ground truth}, as it represents the true state of the system, subject to the given imprecision. Together, $Q$ and $\beta$ define a set $n$ independent random variables $\{X_1, \ldots, X_n\}$, with $X_i$ distributed as $P_i$ (see Figure~\ref{fig:model-notation}).

\begin{remark*}
    Our choice of using a uniform probability distribution is not critical to our approach. We use it for the ease of presentation. In Section~\ref{sec:general-prob} we show that our results apply to a much broader class of distributions.
\end{remark*}

\paragraph*{The Local-Motion Model} 

As mentioned in the introduction, there are many ways to model the realistic motion of a collection of agents in an environment. A natural requirement is that each agent's motion is affected by the presence of nearby agents. While there are many ways to incorporate this information (see, e.g., \cite{SSK05}), we have chosen a very simple model, where velocities are influenced by the distance to the nearest neighbor.

We think of the objects of our system as moving continuously over time, and our algorithm queries their state at regular discrete \emph{time steps}. From the algorithm's perspective, objects move, or ``evolve,'' over time in the following manner. Given a parameter $\alpha$, where $0 < \alpha < 1$, at every time step, an agent called the \emph{evolver}, selects an object $i \in [n]$ and moves $\vctr q_i$ by a distance of at most $\alpha N_i$ in any direction of its choosing (see Figure~\ref{fig:model}). While the value of $\alpha$ is known, the action of the evolver, including the choice of the object and the movement, is hidden from the algorithm. The evolver on the other hand, is a \emph{strong adversary}, meaning it has access to our algorithm, along with the input set~\cite{BoE05}.

For a nonnegative integer $t$, let $Q \spr t$ and $\mc P \spr t$ denote the underlying centers and distributions, respectively, at time $t$. To simplify notation, we omit the superscript when talking about the current time.

We assume that, throughout the process, the points $Q \spr t$ are restricted to lie within a \emph{bounding region}. While the exact size and shape of this region is not important, for concreteness we take it to be a Euclidean ball centered at the origin $O$, which we denote by $\mc B_0$. The algorithm has knowledge of this ball. Define the system's initial \emph{aspect ratio} to be 
\[
    \Lambda_0 
        ~ = ~ \frac{\radius(B_0)}{\min_i N_i \spr 0}.
\]
Given any positive constant $c$, let $c \mc B_0$ denote a factor-$c$ expansion of this ball about the origin.

\begin{figure}[htbp]
    \centering
    \begin{subfigure}[t]{0.5\textwidth}
        \centering
        \includegraphics[scale=0.35,page=1]{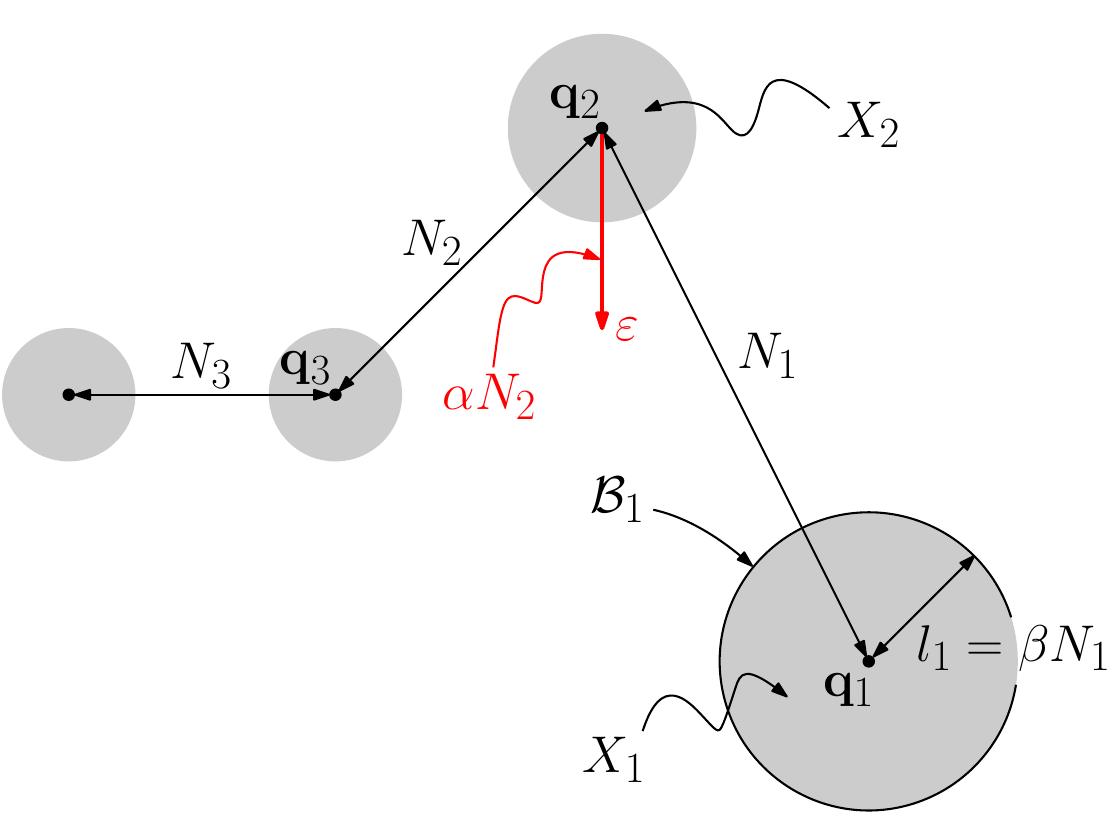}
        \caption{Illustration of the objects, and the notations used. The evolver's action $\eps$ is shown in red.} \label{fig:model-notation}
    \end{subfigure}\hfill
    \begin{subfigure}[t]{0.45\textwidth}
        \centering
        \includegraphics[scale=0.35,page=2]{Figs/model}
        \caption{Objects that were affected by $\eps$ are shown in pink. Tiled regions show the previous state.}
    \end{subfigure}
    \caption{The model and an action by the evolver. Shaded regions represent objects.} \label{fig:model}
\end{figure}

\paragraph*{Oracle}

Rough knowledge about the current state of the system is provided by a radar-like object, an \emph{oracle} $\mc O$, which is given the center $\vctr x$ and radius $r$ of a ball $\mc B(x,r)$, and an object index $i$. The call $\mc O(i, \vctr x, r)$ returns the following regarding the current state of the system:
\begin{itemize}
\item ``$Y$'' or ``$N$'' depending on whether $X_i \in \mc B(\vctr x,r)$
\item ``$+$'' or ``$-$'' depending on whether there exists $j \neq i$, such that $X_j \in \mc B(\vctr x,r)$
\end{itemize}
The first element of the pair is used to estimate the location of object $i$, and the second is employed to estimate the distance to its nearest neighbor. Succinctly, $\mc O: [n] \times \RE^d \times \RE^+ \rightarrow \{Y,N\}\times\{+,-\}$. Observe that because $X_i$ and $X_j$ are random variables, so is $\mc O(i, \vctr x, r)$. Consistent with prior applications of the evolving data framework (see, e.g., \cite{AKM12}), we have intentionally made the oracle as weak as possible, implying that our algorithm can be used with stronger oracles.

\paragraph*{Hypothesis and Distance}

The algorithm maintains a \emph{hypothesis} of the current object locations, which is defined to be a set $\mc H = \{H_1, \ldots, H_n\}$ of spatial probability distributions in $\RE^d$. The \emph{distance} $\mc D$ of the truth $\mc P$ from the current hypothesis $\mc H$, denoted $\mc D (\mc P, \mc H )$, is defined as the sum of $n$ Kullback-Leibler divergences from the hypothesized distributions to the true ones. For $i \in [n]$, let
\[
    D_i
        ~ = ~ D_{\text {KL}} (P_i \parallel H_i) 
        ~ = ~ \int_{\vctr x \in \mc B_i}P_i(\vctr x) \log \frac{P_i(\vctr x)}{H_i(\vctr x)} \mu(d\vctr x),
\]
where $\mu(\cdot)$ denotes the measure over $\mc B_i = \mc B(\vctr q_i,\beta N_i)$ and define 
\begin{equation}
    \mc D (\mc P, \mc H )
        ~ = ~ \sum _{i=1}^n D_i
        ~ = ~ \sum _{i=1}^n D_{\text {KL}} (P_i \parallel H_i). \label{Defn:Dist}
\end{equation}

Next we consider a simple hypothesis, which assumes no information about the locations of the objects, other than the fact that they lie inside the bounding region. Recall that $\mc B_0$ denotes this bounding region and $3 \mc B_0$ denotes a factor-$3$ expansion about its center. Since all the points of $Q$ lie within $\mc B_0$ and $0 < \beta < 1$, $3 \mc B_0$ is guaranteed to contain all the $\beta$-local feature regions. For all $i \in [n]$, define the \emph{na\"{\i}ve hypothesis}, denoted $H_i^*$, to be the uniform distribution over $3 \mc B_0$. Irrespective of the initial truth, the initial distance of the $i$th object satisfies the following bound.
\begin{align*}
    D_i^*
          ~ = ~ \int_{\vctr x \in \mc B _i}P_i(\vctr x) \log \frac{P_i(\vctr x)}{H_i^*(\vctr x)} \mu(d\vctr x) 
        & ~ = ~ \int_{\vctr x \in \mc B _i}P_i(\vctr x) \log\frac{1/(\beta N_i)^d}{1/(3 \cdot\radius(\mc B_0))^d} \mu(d\vctr x) \\
        & ~ = ~ \log \left( \frac{3 \cdot\radius(\mc B_0)}{\beta N_i} \right)^d \int_{\vctr x \in \mc B _i}P_i(\vctr x)\mu(d\vctr x) \\
        & ~ \leq ~ d \left( \log \Lambda_0 + \log \frac{3}{\beta} \right).
\end{align*}
Under our assumption that $d$ and $\beta$ are constants, we conclude that $\mc D^* \in O(n \log \Lambda_0)$. $\Lambda_0$, the initial aspect ratio, depends on the initial configuration of the points of $Q$, and can be quite bad. We therefore ask whether we can achieve a smaller bound on the distance, or better yet maintain a bound which is independent of the original configuration.

The combination of evolver, oracle, and distance function constitute a model of evolving motion, which we henceforth call the \emph{$(\alpha,\beta)$-local-motion model}.

\paragraph*{Class of Algorithms}

We assume the algorithm that maintains the hypothesis runs in discrete steps over time. With each step, it may query the oracle a constant number of times, perform a constant amount of work, and then update the current set of hypotheses. The number of oracle queries is independent of $n$ but may depend on the dimension $d$, the local-feature scale factor $\beta$, and the motion factor $\alpha$. In our case, this work takes the form of updating the hypothesis for the object that was queried. In the purest form of the evolving data framework, the evolver and algorithm alternate~\cite{AKM11}. Instead, similar to the generalized framework proposed in \cite{AcM22}, we assume that there is a fixed \emph{speedup factor}, denoted $\sigma$. When amortized over the entire run, the ratio of the number of steps taken by the algorithm and the evolver does not exceed $\sigma$.

\paragraph*{Objective and Results}

The objective of the algorithm is as follows. Given any starting ground-truth configuration, after an initial ``burn-in'' period, the algorithm guarantees that the hypothesis is within a bounded distance of the truth subject to model assumptions and the given speedup factor. Our main result is that there exists an algorithm with constant speedup factor $\sigma$ that maintains a distance of $O(n)$ in steady state.

\begin{theorem}{}\label{thm:main}
    Consider a set of $n$ evolving objects in $\RE^d$ under the $(\alpha,\beta)$-local-motion model, for constants $\alpha$ and $\beta$, where $0 < \alpha, \beta < \frac{1}{3}$. There exists an algorithm of constant speedup $\sigma$, and burn-in time $t_0 \in O(n \log \Lambda_0 \log\log \Lambda_0)$ such that for all $t \geq t_0$, this algorithm maintains a distance of $O(n)$.
\end{theorem}

The algorithm and its analysis will be proved in~\cref{sec:Alg}. Given that no algorithm can guarantee an exact match between hypothesis and truth, it is natural to wonder how close this is to optimal. In~\cref{sec:lower-bound}, we will show that it is asymptotically optimal by showing that for any algorithm and any constant speedup factor, there exists an evolver that can force a distance of $\Omega(n)$ in steady state.

\paragraph*{Why the KL Divergence?}

The principal challenge in generalizing the evolving framework from simple 1-dimensional applications like sorting to a multidimensional setting is the lack of an obvious distance measure that captures how close the hypothesis is to current state. Our approach is motivated by an information-theoretic perspective. The KL divergence $D_i = D_{\text {KL}} (P_i \parallel H_i)$ serves as a measure of how different the actual object distribution $P_i$ is from our hypothesis $H_i$. (Note that the KL divergence is asymmetric, which is to be expected, given the asymmetric roles of the truth and our hypothesis.)

As an example of this information-theoretic approach, consider the following application in coding theory and space quantization \cite{CT2012}. The objective is to communicate the location of an imprecise point $X_i$ with any arbitrary resolution $\delta$. From Shannon's source coding theorem \cite{mackay2003}, the theoretical lower bound for the expected number of bits required for communication is given by the \emph{Shannon entropy} \cite{Sha48}
\[
    b_{P_i} ~ = ~ \sum_{C_\delta \in B_i} P_i(C_\delta) \log \frac{1}{P_i(C_\delta)},
\]
where $C_\delta$ is a cell of size $\delta$ in $\RE^d$, and $P(C)$, via a slight abuse of notation, is the probability associated with a cell $C$ using the underlying probability distribution function $P$. This is roughly achieved by a Huffman coding \cite{Huf52} of the space around $q_i$, which intuitively assigns a number of bits proportional to the log of the inverse of the probability measure associated with an event, which in our case is $X_i$ lying within a cell $C_\delta$ of size $\delta$. 

Since $P_i$'s are not available to us, we use a similar strategy of encoding the space via $H_i$, and therefore in expectation end up using roughly 
\[
    b_{H_i} ~ = ~ \sum_{C_\delta \in B_i} P_i(C_\delta) \log \frac{1}{H_i(C_\delta)}
\]
bits. When $\delta$ is arbitrarily small, the difference $(b_{H_i} - b_{P_i}) \sim D_i$. Therefore, $ \mc D (\mc P, \mc H ) = \sum _{i=1}^n D_i$ roughly represents the extra bits we use, in expectation, to represent the location of $X_i$'s individually using just $H_i$'s, compared to the absolute theoretical limit.

\paragraph*{Paper Organization}

The remainder of the paper is organized as follows. In the next section, we present Theorem~\ref{thm:lower-bound}, which provides a lower bound on the distance achievable by any algorithm in our model. In \cref{sec:Alg} we present the algorithm and analyze its steady-state performance. In \cref{sec:general-prob} we discuss possible relaxations to the assumptions made in our model.

\section{Lower Bound} \label{sec:lower-bound}

In this section, we show that maintaining $\dist \in O(n)$ is the best we can hope to achieve, as far as bounding the distance in terms of the input size is concerned. This is established in the following theorem. Due to space limitations, the proof has been deferred to \cref{prf:thm:lower-bound}. 

\begin{theorem} \label{thm:lower-bound}
For any algorithm $\mc A$, there exists a starting configuration $Q \spr 0$ and an evolver (with knowledge of $\mc A$) in the local-motion model such that, for any positive integer $t_0$, there exists $t \geq t_0$, such that $\mc D \spr t = \mc D (\mc P \spr t, \mc H \spr t) \in \Omega(n)$.
\end{theorem}

The proof follows a similar structure to other lower-bound proofs in the evolving data framework~\cite{AKM12, BDE18a, AcM22}. Intuitively, over a period of time of length $c \SP n$, for $c < 1$, the algorithm can inspect the locations of only a constant fraction of points. During this time, the evolver can move a sufficient number of the uninspected points so that the overall distance increases by $\Omega(n)$.

\section{Algorithm} \label{sec:Alg}

In this section, we present our algorithm and analyze its performance. We begin by defining our hypothesis. For every index $i$, we define two parameters: $h_i \in \RE$, and $\vctr{k}_i = (k_{i,1}, \ldots k_{i,d}) \in \RE^d$. For $\vctr x = (x_1, \ldots x_d) \in \RE^d$ we set the hypothesized probability density for the $i$th point to be a probability density function (PDF) for the $d$-dimensional independent Cauchy distribution centered at $\vctr k _i$, with scaling parameter $h_i$ along each axis \cite{Ver23}, that is,
\begin{equation}
\label{Eqn:HypDef}
    H_i(\vctr{x}) 
        ~ = ~ f_{h_i,\vctr{k}_i}(\vctr x) 
        ~ = ~ \frac{1}{\pi^d h_i ^d} \prod _{j=1}^d \left({1+\frac{\left(x_j - {k}_{i,j}\right)^2}{h_i ^2}}\right)^{-1} \eqtag{Hypothesis Def.}
\end{equation}   
If we let $f_{1, \vctr{0}}(\vctr x)$ denote the PDF for the standard $d$-dimensional independent Cauchy distribution (centered at the origin unit scale), we can express this equivalently as
\[
    f_{h_i,\vctr{k}_i}(\vctr x) 
        ~ = ~ \frac{1}{h^d_i} f_{1, \vctr{0}}\left(\frac{\vctr x - \vctr k_i}{h_i}\right), \quad \text{where} \quad
    f_{1, \vctr{0}}(\vctr x) 
        ~ = ~ \frac{1}{\pi^d} \prod _{j=1}^d \left(1+x_j^2\right)^{-1}.
\]

It will be convenient to define the \emph{hypothesis ball} $\mc B^H_i = \mc B (\vctr k_i, h_i)$, which we use to illustrate $H_i$. Note however, that the Cauchy PDF is nonzero over all of $\RE^d$. This is in contrast to the ground-truth PDFs, $P_i$, which have bounded support over $\mc B_i$. Our algorithm modifies the parameters $\vctr k_i$ and $h_i$ in response to information received from the oracle.

\begin{remark*}
    As mentioned above, the truth $P_i$'s have a bounded support and $H_i$'s are defined over all of $\RE^d$. This in stark contrast to the purest form of the evolving framework \cite{AKM11}, where the hypothesis and the truth have identical structures. We relax the framework and note that there is no reason for both of them to have the same structure, as long as we have a notion of distance between them.
\end{remark*}

\subsection{Potential Function} \label{sec:potential}

Due to the subtleties of tracking the Kullback-Leibler divergence, we introduce a potential function $\Phi$ to aid in the analysis. For each object index $i \in [n]$, we define an individual potential function $\Phi_i$, which bounds the distance $D_i$ to within a constant. The overall potential $\Phi$ will be the sum of these functions. 

Before exactly defining $\Phi_i$, let us observe a few things about $H_i$. Let $s_i = \|\vctr q_i - \vctr k_i\|$ be the Euclidean distance between the center of the actual distribution and the center of the hypothesis ball. Let $l_i = \beta N_i$ denote the radius of the local feature of $\vctr q_i$. Also recall that $h_i$ is the radius of the hypothesis ball (see \cref{fig:Potential_notation}). We define the individual and system \emph{potentials} to be
\begin{equation}
    \Phi_i 
        ~ = ~ \log \left(\frac{\max (s_i,l_i,h_i)}{\sqrt{l_i h_i}}\right) 
        \quad \text{and} \quad 
    \Phi ~ = ~ \sum_{i \in [n]}\Phi_i \eqtag{Potential Definition}\label{Eqn:PotDef}
\end{equation}

\begin{figure}[htbp]
    \centering
    \includegraphics[scale = 0.45]{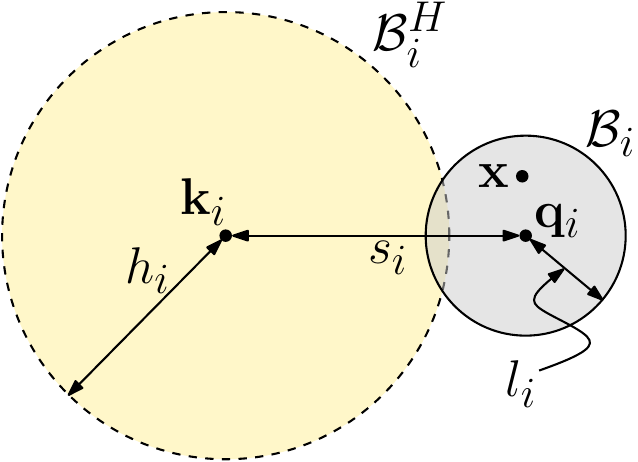}
    \caption{The definition of the individual potential $\Phi_i$.} \label{fig:Potential_notation}
\end{figure}

Let us explore the relationship between the potential and distance. Recall that $P_i(\vctr x)$ is the probability density function of a uniform distribution over a ball of radius $l_i$. Thus, for $\vctr x\in \mc B_i$, $P_i(\vctr x) = \omega_d/l_i^d$, where $\omega_d$ denotes the volume of the unit Euclidean ball in $\RE^d$, and $\mc B_i$. Therefore,
\begin{flalign}
    D_i 
        & ~ = ~ \int _{\vctr x \in \mc B _i} P_i(\vctr x) \log \frac{P_i(\vctr x)}{H_i(\vctr x)} \mu(d\vctr x) \nonumber\\ 
        & ~ = ~ \int _{\vctr x \in \mc B _i} P_i(\vctr x) \log \frac{\omega_d/l_i^d}{(\pi^d h_i ^d)^{-1} \prod _{j=1}^d \left({1+\frac{\left(x_j - {k}_{i,j}\right)^2}{h_i ^2}}\right)^{-1}} \mu(d\vctr x) \nonumber\\
        & ~ = ~ \int _{\vctr x \in \mc B _i} P_i(\vctr x) \log \left(\omega_d \SP \pi^d\right) \mu(d\vctr x) ~ + ~ 
            \int _{\vctr x \in \mc B _i}P_i(\vctr x) \log\frac{h_i^d ~ \prod _{j=1}^d \left({1+\frac{\left( x_j - {k}_{i,j}\right)^2}{h_i ^2}}\right)}{l_i^d} \mu(d\vctr x). \label{Eqn:DistInit}
\end{flalign}
For $\vctr x \in \mc B _i$ and $j \in [d]$, by the triangle inequality, we have
\[
    |x_j - {k}_{i,j}| 
        ~ \leq ~ \|\vctr x - \vctr k_i\| 
        ~ \leq ~ \|\vctr x - \vctr q_i\| + \|\vctr q_i - \vctr k_i\| 
        ~ \leq ~ l_i + s_i.
\]
Therefore,
\begin{flalign*}
    D_i
        & ~ \leq ~ \log\left(\omega_d \SP \pi^d\right) ~ + ~ 
            \int _{\vctr x \in \mc B _i} P_i(\vctr x) \log \left[\frac{h_i^d}{l_i^d} \SP \prod _{j=1}^d \left( 1 + \frac{\left(s_i + l_i\right)^2}{h_i ^2} \right) \right] \mu(d\vctr x) \\
        & ~ \leq ~ \log\left(\omega_d \SP\pi^d\right) ~ + ~ 
            \log \left[\frac{h_i^d}{l_i^d} \left( 1 + \frac{\left(s_i +l_i\right)^2}{h_i ^2} \right)^{\kern-2pt d} \right] \int _{\vctr x \in \mc B _i} P_i(\vctr x)  \mu(d\vctr x) \\
        & ~ = ~ \log\left(\omega_d~ \pi^d\right) ~ + ~
            d \log \left(\frac{h_i^2 + \left(s_i +l_i\right)^2}{l_i ~ h_i} \right)
          ~ \in ~ O(1) + O\left( \frac{\max (s_i,l_i,h_i)}{\sqrt{l_i h_i}} \right).
\end{flalign*}

In summary, we have the following relationship between distance and potential.

\begin{lemma} \label{lem:PotDist}
    Given the potential functions $\Phi_i$ and $\Phi$ defined above, $D_i \in O(1) + O(\Phi_i)$ and $\mc D = \sum D_i \in O(n) + O(\Phi)$.
\end{lemma}
\begin{remark*}[Potential function and approximating pairwise distances]
    The potential function $\Phi_i$ we define here is symmetric. It is remarkable that a symmetric function bounds an asymmetric function $D_i$ under the choice of Cauchy distributions as hypotheses. 

    To demonstrate the value of our potential function, suppose that for $\forall i, \Phi_i < \log c$, for some constant $c$. This implies that $h_i < c\sqrt{l_ih_i}$, which further implies $h_i < c'l_i$ for some constant $c'$. Similarly, $s_i < c\sqrt{l_i h_i} < c'' l_i$ for some constant $c''$. Now, by the triangle inequality, the distance between the centers of two hypothesis balls satisfies
    \[
        \|\vctr k_i - \vctr k_j\| 
            ~ \leq ~ s_i + s_j + \|\vctr q_i - \vctr q_j\| 
            ~ \leq ~ c''l_i + c''l_j + \|\vctr q_i - \vctr q_j\|.
    \]
    Since $l_i = N_i/\beta \leq \|\vctr q_i - \vctr q_j\|/\beta$, we have $\|\vctr k_i - \vctr k_j\| \leq c^*\|\vctr q_i - \vctr q_j\|$, for some constant $c^*$.

    Since $\|\vctr k_i - \vctr k_j\|$ is a constant approximation for $\|\vctr q_i - \vctr q_j\|$, it immediately follows that we can maintain a constant-weight approximation of structures like the Euclidean minimum spanning tree and the Euclidean traveling salesman tour on $Q$ by constructing the same on $\vctr k_i$'s, the centers of the hypotheses $H_i$'s. \hfill \qedsymbol
\end{remark*}

An important feature of our choice of a potential function is that the evolver cannot change its value by more than a constant with each step. Recall that the evolver selects an object index $i$ and moves $q_i$ by at a distance of most $\alpha N_i$ in any direction. This results in at most an $\alpha$ fraction change in the values of $l_i$ and $s_i$, and hence the resulting change in $\Phi_i$ is bounded by a constant. Note, however that the movement of $q_i$ could potentially affect the local feature sizes of a number of other objects in the system. We can show via a packing argument that there are only a constant number of indices $j$, whose $\Phi_j$ value is affected. Furthermore, these values are changed by only a constant. This is encapsulated in the following lemma. Its proof has been deferred to \cref{prf:lem:EvoStep}.

\begin{lemma} \label{lem:EvoStep}
    Each step of the evolver increases the potential $\Phi$ by at most a constant.
\end{lemma}

\subsection{The Algorithm}

Next, we present our algorithm. The algorithm runs in parallel to the evolver, running faster by a speedup factor of $\sigma$ (whose value will be derived in our analysis). Each \emph{step} of the algorithm involves a probe of an object by the oracle, and following that, a possible modification of a hypothesis, $H_i$. For ease of analysis we assume the evolver takes at most one step every time unit, and that the number of steps taken by the algorithm over any time interval of length $t$ is $\sigma t$.

Let us begin with a high-level explanation.  The details are provided in \cref{Alg:TrackByZoom} in the appendix. Recall that we are given a set of $n$ objects, each represented by $X_i$, a uniform probability distribution over a ball $\mc B_i = \mc B(\vctr q_i, \beta N_i)$, where $N_i$ is the distance to its nearest neighbor and $\beta$ is the local feature scale. The points $\vctr q_i$ are constrained to lie within a bounding ball $\mc B_0$ centered at the origin. The algorithm maintains a hypothesis in the form of $n$ Cauchy distributions, each represented by a hypothesis ball $\mc B^H_i = \mc B(\vctr k _i, h_i)$, where $\vctr k_i$ is the center of the distribution and $h_i$ is the distribution's scaling parameter.

The algorithm begins with an initial hypothesis, where each hypothesis ball is just $\mc B_0$. This reflects the fact that we effectively assume nothing about the location of $\vctr q_i$, except that it lies within the bounding ball. The algorithm then proceeds in a series of \emph{iterations}, where each iteration handles all the $n$ indices in order. The handling of each index involves two processes, called \emph{zoom-out} and \emph{zoom-in}. 

%
\begin{figure}[htbp]
    \centering
    \begin{subfigure}[t]{0.6\textwidth}
        \centering
        \includegraphics[scale=0.35]{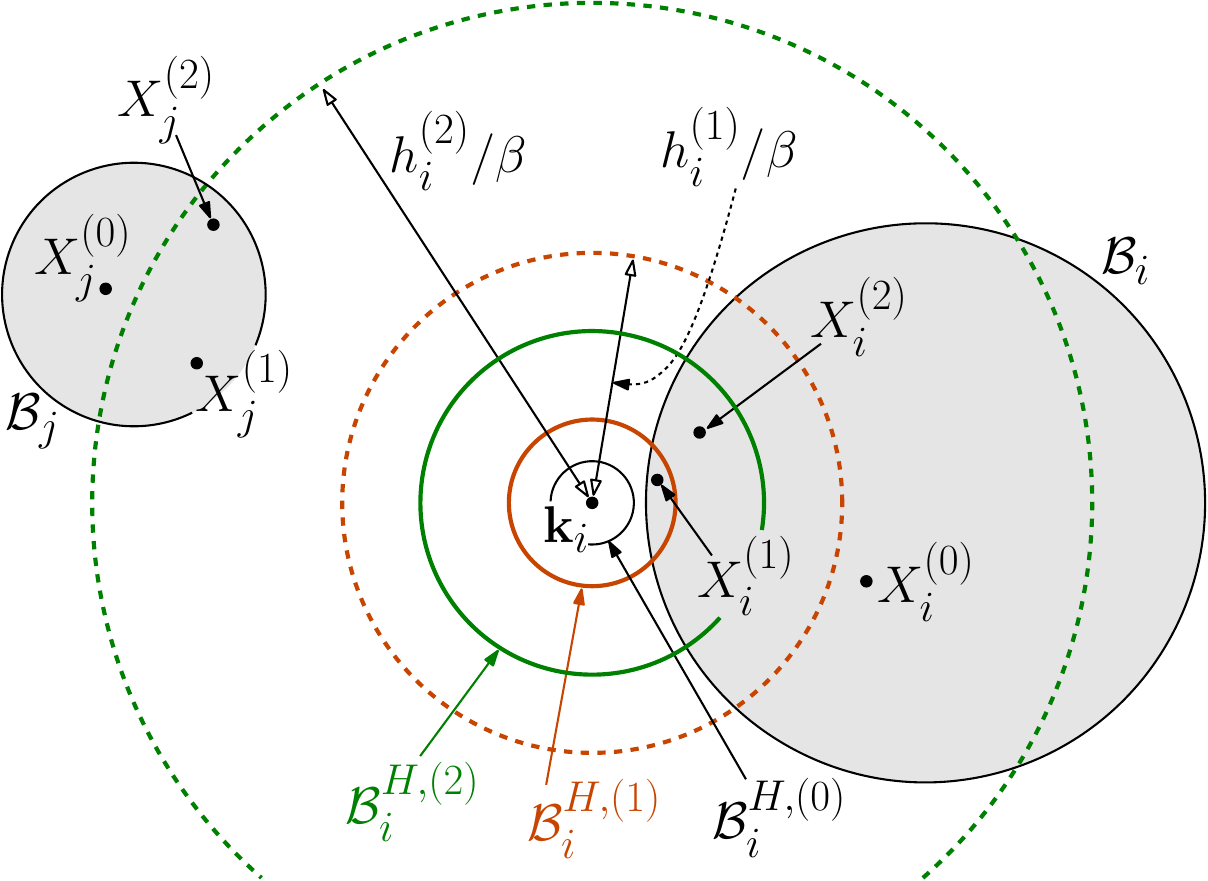}
        \caption{Multiple stages of the zoom-out process starting at $t=0$. Stage $t=2$ is the first time $\mc B^{H,(t)}_i$ contains $X_i$, and its $1/\beta$-expansion contains $X_j$.}
    \end{subfigure}\hfill
    \begin{subfigure}[t]{0.35\textwidth}
        \centering
        \includegraphics[scale=0.35]{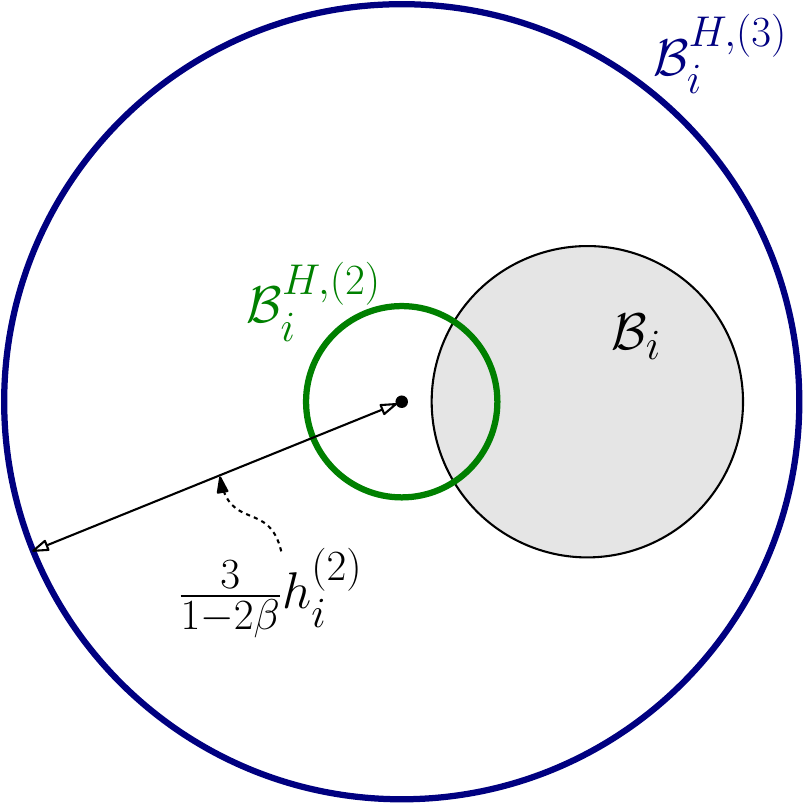}
        \caption{$\mc B_i^{H,(3)}$ is the hypothesis ball at the end of the zoom-out process. Note that $\mc B_i \subset \mc B_i^{H,(3)}$.}
    \end{subfigure}
    \caption{The zoom-out process of \ref{Alg:TrackByZoom} (Line \ref{alg:line:zoomout}).} \label{fig:zoom-out}
\end{figure}

In the zoom-out process, we query the oracle on index $i$ to check whether (1) the sampled point $X_i$ lies within its hypothesis ball and (2) at least one other $X_j$ lies within a concentric ball whose radius is larger by a factor of $\frac{1}{\beta}$. If so, we expand the hypothesis ball by a factor of $3/(1 - 2\beta)$. (We show in Lemma~\ref{lem:exp-factor} that this guarantees that the hypothesis ball $\mc B_i^H$ now contains the local feature ball $\mc B_i$.) We then proceed to the zoom-in process. If not, we double the hypothesis ball and repeat (see \cref{fig:zoom-out}).  

\begin{figure}[htbp]
    \centering
    \begin{subfigure}[t]{0.5\textwidth}
        \centering
        \includegraphics[scale=0.35]{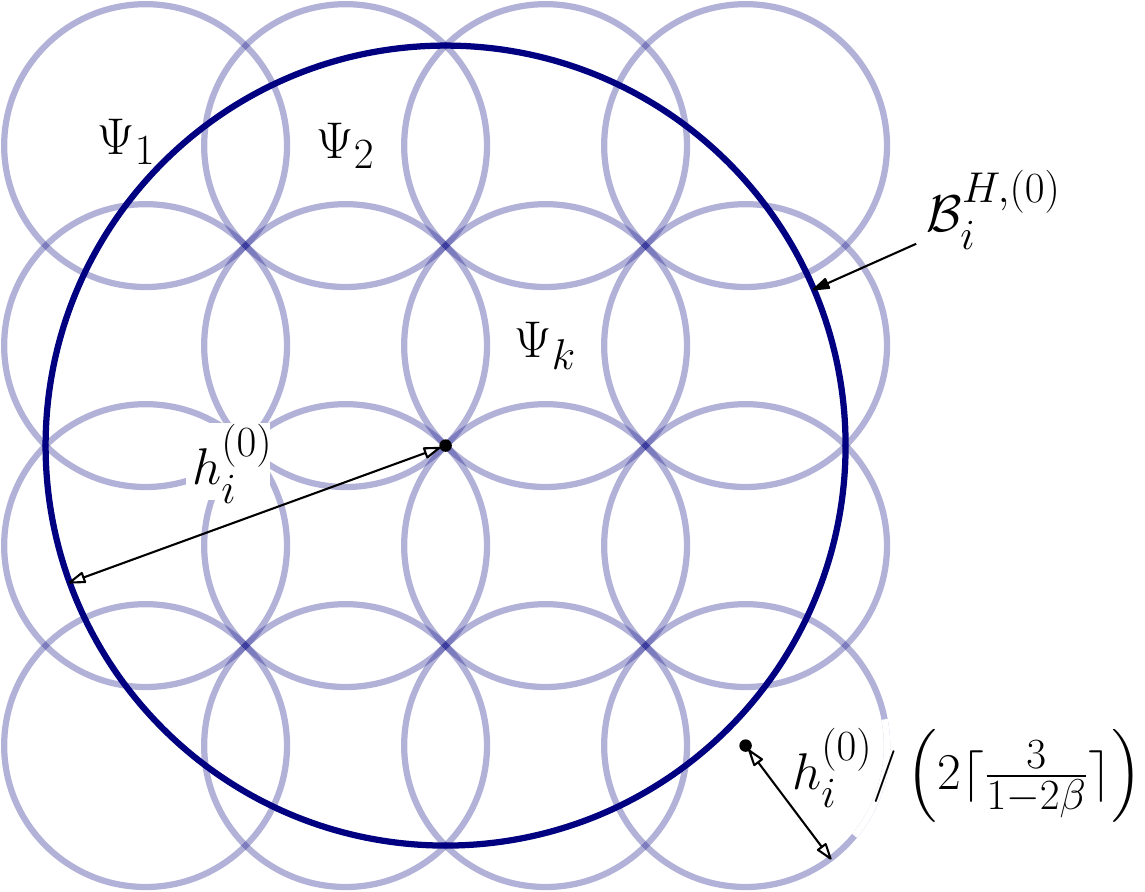}
        \caption{Set of nested balls covering $\mc B_i^{H,(0)}$, which we call $\Psi$. $\radius(\Psi_i) = \radius(\mc B_i)/\left(2\lceil \frac{3}{1-2\beta} \rceil\right)$, and $|\Psi|\in O(1)$.} \label{fig:nested-balls}
    \end{subfigure}\hfill
    \begin{subfigure}[t]{0.45\textwidth}
        \centering
        \includegraphics[scale=0.35]{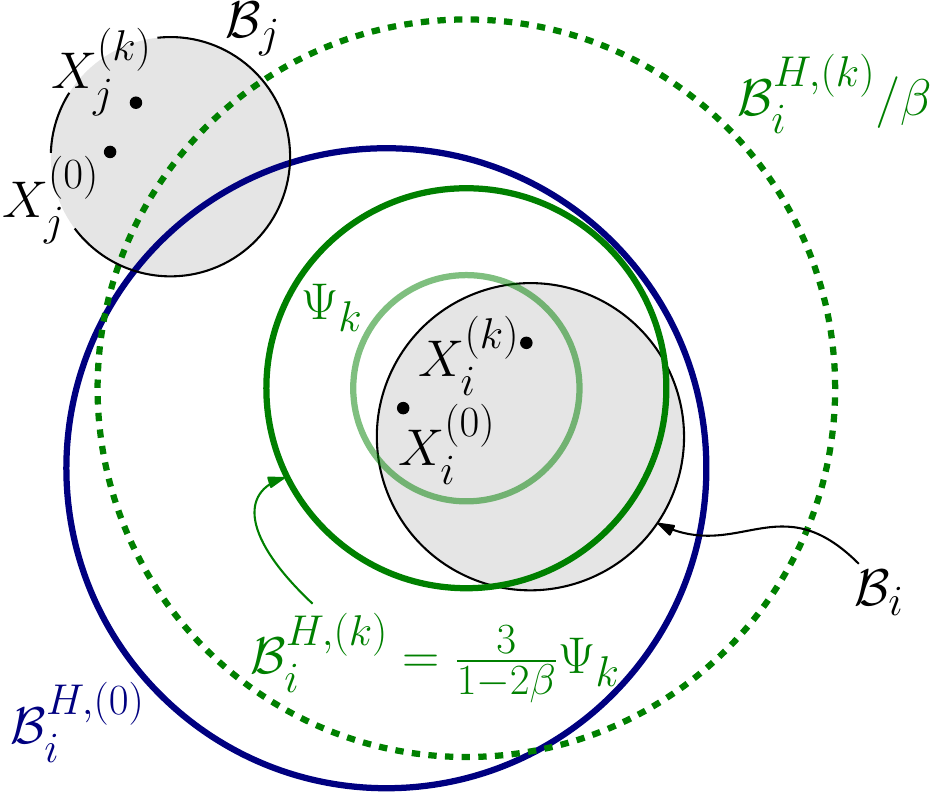}
        \caption{$\Psi_k$ is the first nested ball to contain $X_i$. $\mc B^{H,(k)}_i$ contains $X_i$, and its $1/\beta$-expansion does not contain $X_j$. \ref{Alg:TrackByZoom} moves on to index $i+1$.}
    \end{subfigure}
    \caption{Illustration of the zoom-in process in \ref{Alg:TrackByZoom} (Line \ref{alg:line:zoomin}).} \label{fig:zoom-in}
\end{figure}

In the zoom-in process, we first check whether $X_i$ lies within the hypothesis ball. (The evolver could have moved $q_i$.) If not, we return to the zoom-out process. If so, we next check the $\frac{1}{\beta}$-expansion of this ball. If there is no other $X_j$ in this expanded ball, we accept the current hypothesis for $i$, and go on to the next point in the set. (We show in \cref{lem:almst-trckd} that $\Phi_i$ at this point in time is bounded by a constant). If there is such an $X_j$ however, we may need to shrink the hypothesis ball for the current index. We cover the hypothesis ball with $O(1)$ balls whose radii are smaller by a constant factor, and we test whether $X_i$ lies within any of them (see \cref{fig:zoom-in}). As soon as we find one, we shrink the hypothesis ball about this ball. We repeat this process until one of the earlier conditions applies (causing us to return to the zoom-out process or move on to the next point).


\subsection{Analysis}

In this section, we analyze the distance that the algorithm maintains with the ground truth. Recall that the initial hypotheses are based on the bounding ball $\mc B_0$, which is centered at the origin. Letting $R_0 = \radius(\mc B_0)$, we have $H_i \spr 0 (\vctr x) = f_{\vctr 0,R_0}(\vctr x)$, as defined in Eq.~\eqref{Eqn:HypDef}. Using Eq.~\eqref{Eqn:PotDef}, the initial potential function satisfies:
\begin{equation}\label{eqn:pot-zero}
    \Phi_i \spr 0 = \log \sqrt{R_0/l_i \spr 0}, \quad \text{and} \quad \Phi \spr 0 \in O(n \log \Lambda_0) \quad\quad (\Lambda_0\text{:Aspect Ratio})
\end{equation}



Let us consider how \cref{Alg:TrackByZoom} affect the potential. For a particular index $i$, during the \emph{zoom-out} process, we double the radius $h_i$ of the hypothesis ball. If the Euclidean distance between the centers of the hypothesis and the local-feature ball is much larger than $h_i$, then $\Phi = \log(\max(s_i,l_i,h_i)/\sqrt{l_ih_i}) = \log(s_i/\sqrt{l_ih_i})$ decreases by an additive constant. However, once $h_i$ is larger than $s_i$, the algorithm risks increasing $\Phi$. We show that number of steps where $\Phi$ increases is a constant for every index $i$. During the zoom-in process however, we maintain the invariant that the hypothesis ball contains the local-feature ball (\cref{lem:exp-factor}), implying that $\max(s_i,l_i,h_i) = h_i$. Therefore, $\Phi$ decreases by a constant amount whenever we shrink the hypothesis ball by a constant factor. However, if the evolver moves $q_i$ while index $i$ is being processed, our algorithm may transition to the zoom-out process again. Hence, if the evolver executes some $\eps \itr z$ steps in the $z$th iteration, our algorithm may increase $\Phi$ by only $O\left(n+\eps \itr z\right)$. We obtain the following lemma. The proof is deferred to  \cref{prf:lem:alg-steps}

\begin{lemma}[Algorithm’s contribution towards $\Phi$]\label{lem:alg-steps}
    For any iteration $z$, let $\eps \itr z$ denote the number steps executed by the evolver. Then in the $z$th iteration \ref{Alg:TrackByZoom} increases $\Phi$ in $O\left(n+\eps \itr z\right)$ steps, each time by $O(1)$. In each of the remaining steps of the $z$th iteration it decreases $\Phi$ by $\Theta(1)$ in an amortized sense.
\end{lemma}

Let $\Phi \itr z$ be the potential function at the start of the $z$th iteration. And let $\Delta \itr z$ be the time taken by the $z$th iteration of the algorithm. At the conclusion of the \emph{zoom-in} process for a particular index $i$, the algorithm fixes a hypothesis $H_i$ which is at a reasonable distance from the truth $P_i$. In fact, we show in \cref{lem:almst-trckd} that at this point in time, $\Phi_i$ is $\phi_0$, a constant. Therefore by \cref{lem:alg-steps}, \ref{Alg:TrackByZoom} running at a constant speedup $\sigma$ roughly spends $\sim \Phi \itr z _i$ time (within constant factors), in reducing $\Phi \itr z _i$ to $\phi_0$. Summing up over all indices, we see that the time taken by $z$th iteration, $\Delta \itr z$, is $\sim \Phi \itr z$. Intuitively, for a large enough speedup, the actions of the evolver can only affect $\Delta \itr z$ by a small fraction of this amount. We summarize these observations in the following lemma. The proof is deferred to \cref{prf:lem:iter-time-pot}

\begin{lemma}[Iteration time proportional to Potential]\label{lem:iter-time-pot}
    There exists a constant speedup factor $\sigma$ for \ref{Alg:TrackByZoom} such that $\Delta \itr z \in \Theta_\sigma\left(n+\Phi \itr z\right)$.
\end{lemma}

We are now ready to derive how $\Phi \itr z$ changes over the course of multiple iterations. We show that for a sufficiently large (constant) speedup $\sigma$, after $\sim \log \log \Lambda_0$ iterations $\Phi$ converges to $O(n)$. (Recall that $\Lambda_0$ is the initial aspect ratio.)

\begin{lemma}\label{lem:Pot-converge}
    There exists a constant speedup factor $\sigma$ for \cref{Alg:TrackByZoom} and $z_0 \in \ZZ$, such that for all $z > z_0$, $\Phi \itr z \in O(n)$
\end{lemma}
\begin{proof}
    By \cref{lem:EvoStep}, the $z$th iteration the evolver takes at most $\eps \itr z  = O(\Delta \itr z)$ steps and increases $\Phi$ by at most $O(\Delta \itr z)$. Therefore, by \cref{lem:alg-steps}, \ref{Alg:TrackByZoom} increases $\Phi$ by $O(n+\Delta \itr z)$. Since the algorithm runs at a speedup factor $\sigma$ it takes $\sigma \Delta \itr z$ steps, and by \cref{lem:alg-steps}, it decreases $\Phi$ by at least $\lambda (\sigma~ \Delta \itr z - O(n+\Delta \itr z))$, for some $\lambda \in O(1)$. Accounting for all the increases and decreases we have:
    \begin{align*}
        \Phi \itr {z+1} 
            & ~ \leq ~ \Phi \itr z + O(\Delta \itr z) + O(n+\Delta \itr z) - \lambda (\sigma~ \Delta \itr z - O(n+\Delta \itr z))\\
            & ~ \leq ~ \Phi \itr z + O(\Delta \itr z) - \sigma\lambda \Delta \itr z + O(n).
    \end{align*}
    There exists a sufficiently large $\sigma \in O(1)$ such that $\sigma \lambda \Delta \itr z - O(\Delta \itr z) \geq a \Delta \itr z$, for $0 < a \in O(1)$. So,
    \[
        \Phi \itr {z+1} 
            ~ \leq ~ \Phi \itr z  - a\Delta \itr z + O(n).
    \]
    By \cref{lem:iter-time-pot}, for a sufficiently large (constant) $\sigma$, there exists $c_\sigma > 0$ such that $\Delta \itr z \geq c_\sigma \left(n+ \Phi \itr z\right)$. Therefore,
    \begin{align*}
        \Phi \itr {z+1} 
            & ~ \leq ~ (1-a c_\sigma)\Phi \itr z + (1-a c_\sigma) O(n)\\
            & ~ \leq ~ (1-a c_\sigma)^{z+1}\Phi \itr 0 + \left(\sum_{y=1}^{z+1}(1-a c_\sigma)^y\right) O(n).
    \end{align*}
    By Eq.~\eqref{eqn:pot-zero}, $\Phi \itr 0 = \Phi \spr 0 \in O( n \log \Lambda_0)$, and since there exists $b > 1$ such that $(1 - a c_\sigma) < 1/b < 1$, we have
    \begin{align}
        \Phi \itr {z+1}\label{eq:conv-after-burn}
            & ~ \leq ~ \frac{O( n \log \Lambda_0)}{b^{z+1}}  ~+~ O(n)\\ 
            & ~ \leq ~ O(n), \quad\qquad \text{where}~ z > \log \log \Lambda_0. \notag
    \end{align}
\end{proof}

If $\Phi \itr {z} \in O(n)$, then by \cref{lem:iter-time-pot}, $\Delta \itr z \in \Theta(n)$ as well. Thus, $\Phi$ increases by at most $O(n)$ throughout iteration $z$. Therefore, for any time $t$ during iteration $z$, $\Phi \spr t \in O(n)$ as well.

Using Eq.~\eqref{eq:conv-after-burn} we observe that for all $z > 1$, $\Phi \itr z \leq O(n\log \Lambda_0)$. By \cref{lem:iter-time-pot}, this implies that $\Delta \itr z \in \Theta(n \log \Lambda_0)$. Hence, for some $t_0 \in O(n \log \Lambda_0\cdot\log\log \Lambda_0)$, we have $\Phi \spr {t} \in O(n)$, for all $t > t_0$.

By combining Lemmas~\ref{lem:Pot-converge} and~\ref{lem:PotDist}, we complete the proof of \cref{thm:main}. It is noteworthy that there is a trade-off between the speedup factor $\sigma$ and the motion parameter $\alpha$. Recall that with each step, the evolver can move $q_i$ by a distance of up to $\alpha N_i$. If we reduce this parameter to $\alpha' = (1+\alpha)^\frac{1}{c}-1$, then it takes the evolver at least $c$ steps to effect the same change as before. Thus, by reducing the local-motion factor, it is possible to achieve a speed-up factor of $\sigma = 1$. Formally we have,

\begin{corollary}
    There exists a motion parameter $\alpha' < 1$ such that for a set of $n$ evolving objects in $\RE^d$ under the $(\alpha',\beta)$-local-motion model, where $\beta < \frac{1}{3}$, \cref{Alg:TrackByZoom} (with speedup factor 1) maintains a distance of $O(n)$ for all $t \geq t_0$, $t_0 \in O(n \log \Lambda_0 \log\log \Lambda_0)$
\end{corollary}

\section{Extensions and Conclusion} \label{sec:general-prob}


The local-motion model introduced in Section~\ref{sec:problem} assumes a uniform probability distribution for each object. However, the algorithm \ref{Alg:TrackByZoom} with an appropriate speedup factor can be applied to a wider class of probability distributions. A distribution $P_i$ in this class has $\mc B_i$, the volume of the $i$th object, as its support, and is scaled according to $l_i$, the local feature size. Many natural truncated distributions \cite{Dod03}, such as the uniform distribution, the truncated normal distribution \cite{Nie22}, the truncated Cauchy distribution belong to this family. Details are presented in \cref{sec:Extensions}

An additional refinement to the motion model we suggested might entail redefining the local-feature and the speed of a point based on the distance to its $k$-nearest neighbors, for a fixed $k$, instead of just using the nearest neighbor distance. We conjecture that our principal results extend to the generalization as well, given a suitable range counting oracle. Yet another variation of the model could possibly involve letting the uncertainty associated with a point have unbounded support (for e.g. a normal rather than a truncated normal distribution). This is a significantly harder problem, as a point may lie outside its local feature with a constant probability, requiring a maintaining algorithm to have a super-constant speed-up. We aim to study these in the future.




\bibliographystyle{plainurl} 
\bibliography{shortcuts,local-motion}

\begin{thebibliography}{10}

\bibitem{AcM22}
A.~Acharya and D.~M. Mount.
\newblock Optimally tracking labels on an evolving tree.
\newblock In {\em Proc.\ 34th Canad.\ Conf.\ Comput.\ Geom.}, pages 1--8, 2022.

\bibitem{AKM11}
A.~Anagnostopoulos, R.~Kumar, M.~Mahdian, and E.~Upfal.
\newblock Sorting and selection on dynamic data.
\newblock {\em Theoretical Computer Science}, 412(24):2564--2576, 2011.
\newblock \href {https://doi.org/10.1016/j.tcs.2010.10.003} {\path{doi:10.1016/j.tcs.2010.10.003}}.

\bibitem{AKM12}
A.~Anagnostopoulos, R.~Kumar, M.~Mahdian, E.~Upfal, and F.~Vandin.
\newblock Algorithms on evolving graphs.
\newblock In {\em Proc.\ 3rd Innov.\ Theor.\ Comput.\ Sci.\ Conf.}, pages 149--160, 2012.
\newblock \href {https://doi.org/10.1145/2090236.2090249} {\path{doi:10.1145/2090236.2090249}}.

\bibitem{BS23}
Mira Beermann and Anna Sieben.
\newblock The connection between stress, density, and speed in crowds.
\newblock {\em Scientific Reports}, 13(1):13626, 2023.

\bibitem{BDE18b}
J.~J. Besa, W.~E. Devanny, D.~Eppstein, M.~T. Goodrich, and T.~Johnson.
\newblock Optimally sorting evolving data.
\newblock In {\em 45th Internat. Colloq. Autom., Lang., and Prog.}, pages 81:1--81:13, 2018.
\newblock \href {https://doi.org/10.4230/LIPIcs.ICALP.2018.81} {\path{doi:10.4230/LIPIcs.ICALP.2018.81}}.

\bibitem{BDE18a}
J.~J. Besa, W.~E. Devanny, D.~Eppstein, M.~T. Goodrich, and T.~Johnson.
\newblock Quadratic time algorithms appear to be optimal for sorting evolving data.
\newblock In {\em Proc.\ 20th Workshop Algorithm Eng.\ and Exp.}, pages 87--96. SIAM, 2018.

\bibitem{BoE05}
A.~Borodin and R.~El-Yaniv.
\newblock {\em Online Computation and Competitive Analysis}.
\newblock Cambridge University Press, 2005.

\bibitem{BLM11}
K.~Buchin, M.~L{\" o}ffler, P.~Morin, and W.~Mulzer.
\newblock Delaunay triangulation of imprecise points simplified and extended.
\newblock {\em Algorithmica}, 61:674--693, 2011.
\newblock \href {https://doi.org/10.1007/s00453-010-9430-0} {\path{doi:10.1007/s00453-010-9430-0}}.

\bibitem{BEK19}
D.~Busto, W.~Evans, and D.~Kirkpatrick.
\newblock Minimizing interference potential among moving entities.
\newblock In {\em Proc.\ 30th Annu.\ ACM-SIAM Sympos.\ Discrete Algorithms}, pages 2400--2418, 2019.
\newblock \href {https://doi.org/10.5555/3310435.3310582} {\path{doi:10.5555/3310435.3310582}}.

\bibitem{CT2012}
T.M. Cover and J.A. Thomas.
\newblock {\em Elements of Information Theory}.
\newblock Wiley, 2012.

\bibitem{Dod03}
Y.~Dodge.
\newblock {\em The Oxford Dictionary of Statistical Terms}.
\newblock Oxford Univ.\ Press, 2003.
\newblock \href {https://doi.org/10.1002/sim.1812} {\path{doi:10.1002/sim.1812}}.

\bibitem{EvK24}
W.~Evans and D.~Kirkpatrick.
\newblock Minimizing query frequency to bound congestion potential for moving entities at a fixed target time.
\newblock {\em Algorithms}, 17:246, 2024.
\newblock \href {https://doi.org/10.3390/a17060246} {\path{doi:10.3390/a17060246}}.

\bibitem{GLW12}
J.~Gudmundsson, P.~Laube, and T.~Wolle.
\newblock Computational movement analysis.
\newblock In W.~Kresse and D.~Danko, editors, {\em Handbook of Geogr.\ Inf.}, pages 725--741. Springer, 2012.
\newblock \href {https://doi.org/10.1007/978-3-540-72680-7_22} {\path{doi:10.1007/978-3-540-72680-7_22}}.

\bibitem{GSS89}
L.~J. Guibas, D.~Salesin, and J.~Stolfi.
\newblock Epsilon geometry: {Building} robust algorithms from imprecise computations.
\newblock In {\em Proc.\ Fifth Annu.\ Sympos.\ Comput.\ Geom.}, pages 208--217, 1989.
\newblock \href {https://doi.org/10.1145/73833.73857} {\path{doi:10.1145/73833.73857}}.

\bibitem{HBJ05}
D.~Helbing, L.~Buzna, A.~Johansson, and T.~Werner.
\newblock Self-organized pedestrian crowd dynamics: Experiments, simulations, and design solutions.
\newblock {\em Transportation Science}, 39:1--24, 2005.
\newblock \href {https://doi.org/10.1287/trsc.1040.0108} {\path{doi:10.1287/trsc.1040.0108}}.

\bibitem{HW17}
Omar Hesham and Gabriel Wainer.
\newblock Context-sensitive personal space for dense crowd simulation.
\newblock In {\em Proceedings of the Symposium on Simulation for Architecture and Urban Design}, pages 1--8, 2017.

\bibitem{HLS17}
Q.~Huang, X.~Liu, X.~Sun, and J.~Zhang.
\newblock Partial sorting problem on evolving data.
\newblock {\em Algorithmica}, 79:960--983, 2017.
\newblock \href {https://doi.org/10.1007/s00453-017-0295-3} {\path{doi:10.1007/s00453-017-0295-3}}.

\bibitem{Huf52}
David~A. Huffman.
\newblock A method for the construction of minimum-redundancy codes.
\newblock {\em Proceedings of the IRE}, 40(9):1098--1101, 1952.
\newblock \href {https://doi.org/10.1109/JRPROC.1952.273898} {\path{doi:10.1109/JRPROC.1952.273898}}.

\bibitem{Jain1987}
Uday Jain.
\newblock Effects of population density and resources on the feeling of crowding and personal space.
\newblock {\em The Journal of social psychology}, 127(3):331--338, 1987.

\bibitem{Kah91a}
S.~Kahan.
\newblock A model for data in motion.
\newblock In {\em Proc.\ 23rd Annu.\ ACM Sympos.\ Theory Comput.}, pages 265--277, 1991.
\newblock \href {https://doi.org/10.1145/103418.103449} {\path{doi:10.1145/103418.103449}}.

\bibitem{KLM16}
V.~Kanade, N.~Leonardos, and F.~Magniez.
\newblock Stable matching with evolving preferences.
\newblock {\em Approximation, Randomization, and Combinatorial Optimization. Algorithms and Techniques}, 2016.

\bibitem{Ken38}
M.~G. Kendall.
\newblock A new measure of rank correlation.
\newblock {\em Biometrika}, 30:81--93, 1938.
\newblock \href {https://doi.org/10.2307/2332226} {\path{doi:10.2307/2332226}}.

\bibitem{KL51}
S.~Kullback and R.~A. Leibler.
\newblock {On Information and Sufficiency}.
\newblock {\em The Annals of Mathematical Statistics}, 22(1):79 -- 86, 1951.
\newblock \href {https://doi.org/10.1214/aoms/1177729694} {\path{doi:10.1214/aoms/1177729694}}.

\bibitem{Lau14}
P.~Laube.
\newblock {\em Computational Movement Analysis}.
\newblock Springer Cham, 2014.
\newblock \href {https://doi.org/10.1007/978-3-319-10268-9} {\path{doi:10.1007/978-3-319-10268-9}}.

\bibitem{LoK10}
M.~L{\" o}ffler and M.~van Kreveld.
\newblock Largest and smallest convex hulls for imprecise points.
\newblock {\em Algorithmica}, 56:235--269, 2010.
\newblock \href {https://doi.org/10.1007/s00453-008-9174-2} {\path{doi:10.1007/s00453-008-9174-2}}.

\bibitem{mackay2003}
D.J.C. MacKay.
\newblock {\em Information Theory, Inference and Learning Algorithms}.
\newblock Cambridge University Press, 2003.
\newblock URL: \url{https://books.google.com/books?id=AKuMj4PN_EMC}.

\bibitem{Nie22}
F.~Nielsen.
\newblock Statistical divergences between densities of truncated exponential families with nested supports: Duo bregman and duo jensen divergences.
\newblock {\em Entropy}, 24:421, 2022.
\newblock \href {https://doi.org/10.3390/e24030421} {\path{doi:10.3390/e24030421}}.

\bibitem{PZT23}
R.~Pan, Z.~Han, T.~Liu, H.~Wang, J.~Huang, and W.~Wang.
\newblock An {RFID} tag movement trajectory tracking method based on multiple {RF} characteristics for electronic vehicle identification {ITS} applications.
\newblock {\em Sensors}, 23:7001, 2023.
\newblock \href {https://doi.org/10.3390/s23157001} {\path{doi:10.3390/s23157001}}.

\bibitem{SSK05}
A.~Seyfried, B.~Steffen, W.~Klingsch, and M.~Boltes.
\newblock The fundamental diagram of pedestrian movement revisited.
\newblock {\em J.\ Stat.\ Mech.\ Theory Exp.}, 2005:P10002, 2005.
\newblock \href {https://doi.org/10.1088/1742-5468/2005/10/P10002} {\path{doi:10.1088/1742-5468/2005/10/P10002}}.

\bibitem{Sha48}
C.~E. Shannon.
\newblock A mathematical theory of communication.
\newblock {\em The Bell System Technical Journal}, 27(4):623--656, 1948.
\newblock \href {https://doi.org/10.1002/j.1538-7305.1948.tb00917.x} {\path{doi:10.1002/j.1538-7305.1948.tb00917.x}}.

\bibitem{TFK10}
S.~M. Tomkiewicz, M.~R. Fuller, J.~G. Kie, and K.~K. Bates.
\newblock Global positioning system and associated technologies in animal behaviour and ecological research.
\newblock {\em Phil.\ Trans.\ R.\ Soc.\ B}, 365:2163--2176, 2010.
\newblock \href {https://doi.org/10.1098/rstb.2010.0090} {\path{doi:10.1098/rstb.2010.0090}}.

\bibitem{KLS18}
M.~van Kreveld, M.~L{\" o}ffler, F.~Staals, and L.~Wiratma.
\newblock A refined definition for groups of moving entities and its computation.
\newblock {\em Internat.\ J.\ Comput.\ Geom.\ Appl.}, 28(2):181--196, 2018.
\newblock \href {https://doi.org/10.4230/LIPIcs.ISAAC.2016.130} {\path{doi:10.4230/LIPIcs.ISAAC.2016.130}}.

\bibitem{Ver23}
S.~Verd{\'u}.
\newblock The {Cauchy} distribution in information theory.
\newblock {\em Entropy}, 25:346, 2023.
\newblock \href {https://doi.org/10.3390/e25020346} {\path{doi:10.3390/e25020346}}.

\bibitem{WLS18}
L.~Wiratma, M.~L{\" o}ffler, and F.~Staals.
\newblock An experimental comparison of two definitions for groups of moving entities.
\newblock In {\em Proc.\ 10th Internat.\ Conf.\ Geogr.\ Inf.\ Sci.}, pages 64:1--64:6, 2018.
\newblock \href {https://doi.org/10.4230/LIPIcs.GIScience.2018.64} {\path{doi:10.4230/LIPIcs.GIScience.2018.64}}.

\bibitem{YLW15}
S.~Yi, H.~Li, and X.~Wang.
\newblock Understanding pedestrian behaviors from stationary crowd groups.
\newblock In {\em Proc.\ IEEE Conf.\ Comput.\ Vis.\ Pattern Recogn.}, pages 3488--3496, 2015.
\newblock \href {https://doi.org/10.1109/CVPR.2015.7298971} {\path{doi:10.1109/CVPR.2015.7298971}}.

\bibitem{ZZL16}
Y.~Zou, G.~Zeng, Y.~Wang, X.~Liu, X.~Sun, J.~Zhang, and Q.~Li.
\newblock Shortest paths on evolving graphs.
\newblock In H.~T. Nguyen and V.~Snasel, editors, {\em Computational Social Networks}, pages 1--13. Springer, 2016.
\newblock \href {https://doi.org/10.1007/978-3-319-42345-6_1} {\path{doi:10.1007/978-3-319-42345-6_1}}.

\end{thebibliography}

\clearpage
\appendix
\section{Appendix} \label{sec:appndx}

\subsection{Extensions: Other Probability Distributions}\label{sec:Extensions}

The local-motion model introduced in Section~\ref{sec:problem} assumes a uniform probability distribution for each object. In this section we show that algorithm \ref{Alg:TrackByZoom} with an appropriate speedup factor, can be applied to a wider class of probability distributions. Let $\mf P$ be any bounded probability distribution restricted to the $d$-dimensional unit ball, $\mc B (\vctr 0, 1)$. That is, there is a constant $\nu$ and $\mf P: \mc B (\vctr 0, 1) \rightarrow [0,\nu)$. The probability distribution associated with $\vctr q_i$ is defined to be
\[
    P_i(\vctr \vctr x) 
        ~ = ~ \frac{1}{l_i^d} ~\mf P \left(\frac{\vctr x - \vctr q_i}{l_i}\right) 
        ~ = ~ \frac{\mf P (\wt{\vctr x})}{l_i^d},
        \quad\quad \wt{\vctr x} = \left(\frac{\vctr x - \vctr q_i}{l_i}\right)
\]
We define $H_i$ as before (Eq.~\eqref{Eqn:HypDef}) using this probability distribution. Similar to Eq.~\eqref{Eqn:DistInit}, the distance for the $i$th object is
\[
    D_i 
        ~ = ~ \int_{\vctr x \in \mc B _i} P_i(\vctr x) \log \frac{P_i(\vctr x)}{H_i(\vctr x)} \mu(d\vctr x) 
        ~ = ~ \int_{\vctr x \in \mc B _i} P_i(\vctr x) \log P_i(\vctr x) - \int _{\vctr x \in \mc B _i}P_i(\vctr x) \log H_i(\vctr x).
\]
For $\vctr x \in \mc B_i$, we have
\[
    H_i(\vctr x) 
        ~ \geq ~ \frac{1}{\pi^d h_i ^d} \prod _{j=1}^d \left( 1 + \frac{(s_i + l_i)^2}{h_i^2} \right)^{-1},
\]
and therefore,
\[
    D_i 
        ~ \leq ~ \int_{\vctr x \in \mc B _i} P_i(\vctr x) \log P_i(\vctr x) \mu(d\vctr x) ~+~ \log \left(\frac{h_i^2 + (s_i+l_i)^2}{\pi h_i}\right)^d.
\]
By the concavity of the log function and Jensen's inequality, we have $\int P_i(\vctr x) \log P_i(\vctr x) \leq \log \int P_i^2(\vctr x)$. Clearly, $P_i^2(\vctr x) = \mf P^2(\wt{\vctr x})/l_i^{2d}$. Thus,
\[
    D_i 
        ~ \leq ~ \log \int_{\vctr x \in \mc B _i} \frac{\mf P^2(\wt{\vctr x})}{l_i^{2d}} \mu(d\vctr x) + \log \left(\frac{h_i^2 + (s_i+l_i)^2}{\pi h_i}\right)^d.
\]
Since $\mf P(\wt{\vctr x}) < \nu$, and $\vol(\mc B_i) = \omega_d l_i^d$, where $\omega_d$ denotes the volume of the unit Euclidean ball in $\RE^d$, we obtain
\[
    D_i
        ~ \leq ~ \log \frac{\omega_d\nu^2}{l_i^d} + \log \left( \frac{h_i^2 + (s_i+l_i)^2}{\pi h_i} \right)^d 
        ~ =    ~ \log \frac{\omega_d\nu^2}{\pi^d} + d \cdot \log \frac{h_i^2 + (s_i+l_i)^2}{h_i l_i}.
\]
Using the same potential function as defined in Eq.~\eqref{Eqn:PotDef}, we have
\[
    D_i 
        ~ \in ~ O(1) + O(\Phi_i).
\]
Since we argued on the potential function $\Phi$ to conclude our main result, Theorem \ref{thm:main} also holds for any probability distribution that satisfies these criteria. Note many natural distributions, such as the uniform distribution, the truncated normal distribution, the truncated Cauchy distribution, belong to this family. (Truncated distribution here refers to the conditional distribution resulting from restricting the domain to $\mc B_i$: the the volume associated with the $i$th object.)

\clearpage
\subsection{Detailed Algorithm}\label{sec:Det-Algo}
\begin{algorithm}[H]
  \renewcommand{\thealgorithm}{\textsc{TrackByZoom}} 
  \caption{\textsc{Tracking Evolving Distributions}} \label{Alg:TrackByZoom}
  \begin{algorithmic}[1]
    \Procedure{TrackByZoom}{$\mc O,n,\beta, \mc B_0$}
    \Statex \LineComment{Maintain hypothesis $\mc H$ by running the procedure at a speedup factor $\sigma$, given the membership Oracle:$\mc O$, the input size:$n$, the local feature parameter:$\beta$, and the the maximum bounding ball centred at origin:$\mc B_0$} 
    \For{$i \leftarrow$ 1 to $n$}
        \Statex \LeftComment{2}{Initially set the hypotheses according to the given bounding ball}
        \State $h_i \leftarrow \radius(\mc B_0)$, $\vctr k_i \leftarrow \vctr 0$ \Comment{Set $H_i(\vctr x) \leftarrow f_{\vctr 0,R_0}(\vctr x)$, as in Eq.(\ref{Eqn:HypDef}})
    \EndFor 
    \State $i \leftarrow 1$ \Comment{Start with the first point}
    \Statex

    \Label \texttt{\textbf{ZOOM\_OUT:}} \label{alg:line:zoomout}
    \State $\mf O_i \leftarrow \mc O(i,\vctr k_i,h_i)$, $\mf O_i^* \leftarrow \mc O(i,\vctr k_i,h_i/\beta)$ \Comment{Store values to use later}
    \If{$\mf O_i ~ = ~ (Y,\cdot)$ \textbf{and} $\mf O_i^* ~ = ~ (Y,+)$}\label{alg:line:exzoc}
    \Statex \LeftComment{1}{If $\mc B^H_i$ (hypothesis ball) contains $X_i$, and its $1/\beta$-expansion contains some $X_j$}
    \State $h_i \leftarrow \frac{3}{1-2\beta}h_i$ \Comment{Update $\mc B^H_i$ so that it contains $\mc B_i$ (local feature of $\vctr q_i$) }\label{alg:line:exzo}
        \State \Goto \texttt{ZOOM\_IN}
    \EndIf
    \State $h_i \leftarrow 2h_i$ \Comment{Zoom out by expanding $\mc B^H_i$}\label{alg:line:kzo}
    \State \Goto \texttt{ZOOM\_OUT} 
    \Statex \LeftComment{1}{Keep zooming out until $X_i \in \mc B^H_i$, and its $1/\beta$-expansion contains another $X_j$}
    \Statex

    \Label \texttt{\textbf{ZOOM\_IN:}} \label{alg:line:zoomin}
    \State $\mf O_i \leftarrow \mc O(i,\vctr k_i,h_i)$, $\mf O_i^* \leftarrow \mc O(i,\vctr k_i,h_i/\beta)$ \Comment{Store values to use later}
    \If{$\mf O_i~ = ~ (N,\cdot)$} \Comment{$X_i$ no longer in $\mc B^H_i$} \label{alg:line:exevzic}
    \State \Goto \texttt{ZOOM\_OUT}
    \EndIf
    
    \If{$\mf O_i~ = ~ (Y,\cdot)$ \textbf{and} $\mf O_i^*~ = ~ (Y,-)$} \label{alg:line:exzic}
    \Statex \LeftComment{1}{$X_i \in \mc B^H_i$, and its $1/\beta$-expansion doesn't contain another $X_j$}
        \State $i\leftarrow 1+i\mod n$ \Comment{$\Phi_i$ is a constant, ready to move to the next point}\label{alg:line:nxpt}
        \State \Goto \texttt{ZOOM\_OUT}
    \EndIf
    
    \If{$\mf O_i~ = ~ (Y,\cdot)$ \textbf{and} $\mf O_i^*~ = ~ (Y,+)$} \label{alg:line:kzic}
    \Statex \LeftComment{1}{If $\mc B^H_i$ (hypothesis ball) contains $X_i$, and its $1/\beta$-expansion contains some $X_j$}
        \For{$\mc B(\vctr x,r) \in \Psi(\mc B^H_i,2\lceil \frac{3}{1-2\beta} \rceil )$} \label{alg:line:zinested}
        \Statex \LeftComment{2}{$\Psi$ is the set of balls of radius $h_i/\left(2\lceil \frac{3}{1-2\beta} \rceil\right)$ which cover $\mc B^H_i$}
            \If{$\mc O(i,\vctr x,r) ~ = ~ (Y, \cdot)$} \Comment{Found the nested ball that contains $X_i$}
            \State $\vctr k_i \leftarrow \vctr x$, $h_i \leftarrow \frac{3}{1-2\beta}r$ \Comment{Expand the nested ball so that it contains $B_i$}\label{alg:line:ziexp}
            \EndIf
        \EndFor
    \EndIf
    \State \Goto \texttt{ZOOM\_IN} 
    \Statex \LeftComment{1}{Keep zooming in until $X_i \in \mc B^H_i$, and its $1/\beta$-expansion contains no other $X_j$}
    \Statex
    
  \EndProcedure
\end{algorithmic}
\end{algorithm}

\subsection{Deferred Proofs}\label{para:App-proofs}

Note that some of the restatements of the lemmas given here provide additional details, which were not present in the original statements. 

The following lemma justifies the expansion factor $\frac{3}{1-2 \beta}$ used in \cref{Alg:TrackByZoom}. This applies when the conditions of Lines~\ref{alg:line:exzoc} and \ref{alg:line:kzic} are satisfied.

\begin{lemma}[Expansion Factor]\label{lem:exp-factor}
If $\mc O(i,\vctr k_i,h_i) = (Y,\cdot)$ and $\mc O(i,\vctr k_i,h_i/\beta) = (Y,+)$, then $\mc B_i \subseteq \mc B\left(\vctr k_i,\frac{3}{1-2\beta} h_i\right)$.
\end{lemma}
\begin{proof}
The fact that $\mc O(i,\vctr k_i,h_i) = (Y,\cdot)$ implies that $\mc B_i$ and $\mc B^H_i = \mc B(\vctr k_i,h_i)$ intersect. Since $\mc O(i,\vctr k_i,h_i/\beta) = (Y,+)$, there exists $j \neq i$ such that $\|X_j-\vctr k_i\| \leq h_i/\beta$ (see Figure~\ref{fig:expfactor}). Since $X_j$ is sampled from $\mc B_i = \mc B(\vctr q_j, l_j)$, we have $\|X_j - \vctr q_j\| \leq l_j$. By the definition of $l_j$ and the triangle inequality, we have 
\[
    l_j
        ~ =    ~ \beta N_j 
        ~ \leq ~ \beta \|\vctr q_i - \vctr q_j\| 
        ~ \leq ~ \beta (\|\vctr q_i - X_j\| + \|X_j - \vctr q_j\|)
        ~ \leq ~ \beta (\|\vctr q_i - X_j\| + l_j),
\]
which implies that $\|X_j - \vctr q_i\| \geq (1/\beta - 1) l_j$. 
    
\begin{figure}[h]
    \centering
    \includegraphics[scale=0.35]{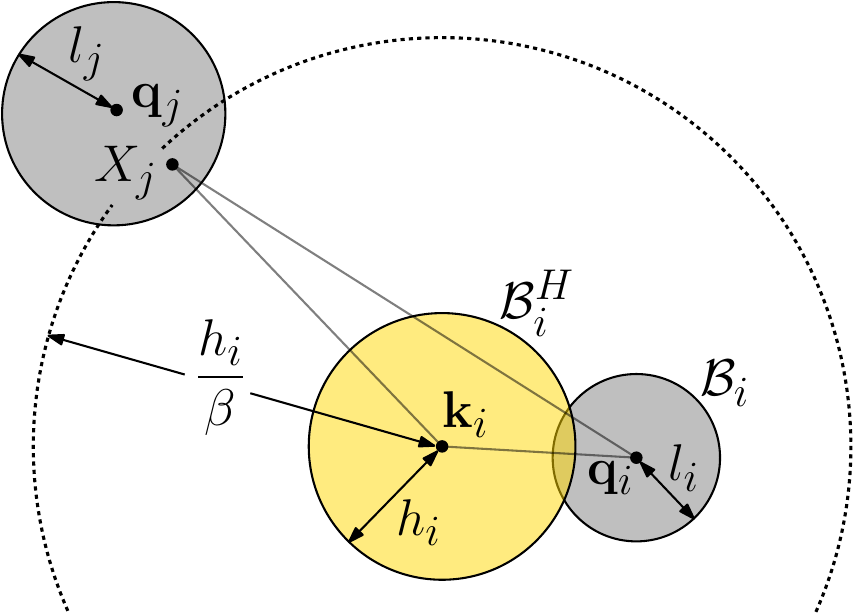}
    \caption{Proof of \cref{lem:exp-factor}.} \label{fig:expfactor}
\end{figure}

We assert that $\|X_j - \vctr q_i\| \geq (1/\beta - 1) l_i$. To see why, first if $l_j \geq l_i$, then Thus,
\[
    \|X_j - \vctr q_i\| 
        ~ \geq ~ \left( \frac{1}{\beta} - 1 \right) l_j 
        ~ \geq ~ \left( \frac{1}{\beta} - 1 \right) l_i.
\]
On the other hand, if $l_i > l_j$, then by the triangle inequality, and the fact that $\|\vctr q_i - \vctr q_j\| \geq N_i = l_i/\beta$. we have
\[
    \|X_j - \vctr q_i\|
        ~ \geq ~ \|\vctr q_i - \vctr q_j\| - \|X_j - \vctr q_j\|
        ~ \geq ~ \frac{l_i}{\beta} - l_j 
        ~ >    ~ \left( \frac{1}{\beta} - 1 \right) l_i.
\] 
By this assertion, the triangle inequality, and our earlier observations, we have, 
\begin{align*}
    l_i \left( \frac{1}{\beta} - 1 \right)
        & ~ \leq ~ \|X_j - \vctr q_i\|
          ~ \leq ~ \|X_j - \vctr k_i\| + \|\vctr k_i - \vctr q_i\| 
          ~ \leq ~ h_i/\beta + (h_i + l_i) \\
        & ~ \leq ~ \left( \frac{1}{\beta} + 1 \right) h_i + l_i,
\end{align*}
which implies that $l_i(1-2\beta) \leq (1+\beta) h_i$, and hence $h_i + 2l_i \leq \frac{3}{1-2\beta}h_i$. Because $\mc B_i$ and $\mc B(\vctr k_i,h_i)$ intersect, we have
\[
    \mc B_i
        ~ \subseteq   ~ \mc B(\vctr k_i, h_i+2l_i)
        ~ \subseteq ~ \mc B\left(\vctr k_i,\frac{3}{1-2\beta}h_i\right),
\] 
as desired.
\end{proof}

The next lemma applies whenever \cref{Alg:TrackByZoom} has completed its processing of an object, on Line~\ref{alg:line:nxpt}. It shows that the potential value for this object does not exceed a constant.

\begin{lemma}[Done tracking $X_i$]\label{lem:almst-trckd}
There exists a constant $\phi_0$, such that whenever \cref{Alg:TrackByZoom} completes its processing of any object $i$, $\Phi_i \leq \phi_0 \in O(1)$.
\end{lemma}

\begin{proof}
The algorithm moves on to the next point when the condition in line \ref{alg:line:exzic} is satisfied ($\mc O(i,\vctr k_i,h_i) = (Y,\cdot)$ and $\mc O(i,\vctr k_i,h_i/\beta) = (Y,-)$). Since $\mc O(i,\vctr k_i,h_i) = (Y,\cdot)$, $\mc B^H_i = \mc B(\vctr k_i,h_i)$ intersect. Moreover, if $\vctr q_j$ is the current nearest neighbor of $\vctr q_i$, $\mc O(i,\vctr k_i,h_i/\beta) = (Y,-)$ implies $X_j$ lies outside the $1/\beta$ expansion of $\mc B^H_i$. See \cref{fig:almst-trckd}.

\begin{figure}[h]
    \centering
    \includegraphics[scale=0.35]{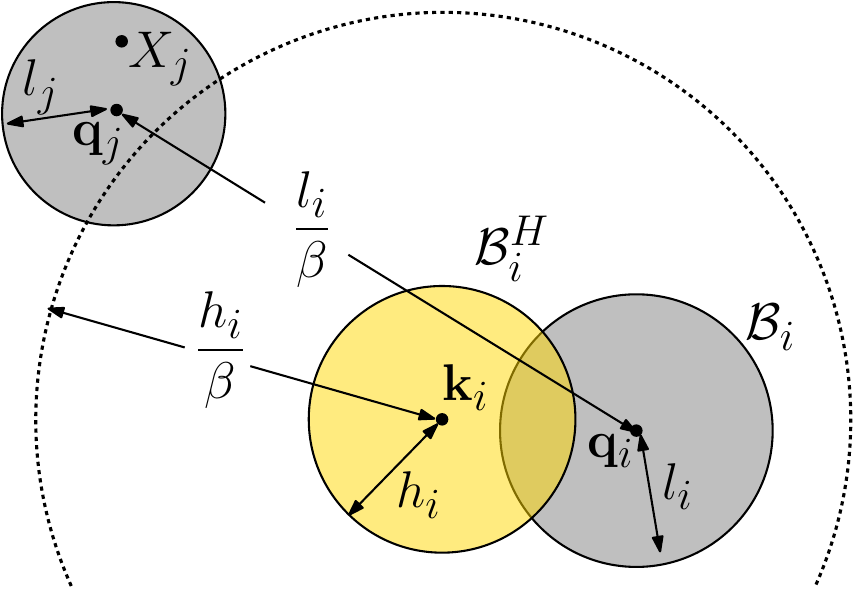}
    \caption{Proof of \cref{lem:almst-trckd}.} \label{fig:almst-trckd}
\end{figure}

Now because $\vctr q_j$ is the nearest neighbor of $\vctr q_i$, $N_j \leq N_i$, hence $l_j \leq l_i$.  By triangle inequality, the following series of inequality holds:
\[
\frac{h_i}{\beta} \leq \|X_j - \vctr k_i\| \leq \|X_j - \vctr q_i\| + \|\vctr q_i - \vctr k_i\| \leq \left(\frac{l_i}{\beta} + l_j\right) + \left(h_i + l_i\right) \leq \left(\frac{1}{\beta} + 2\right)l_i + h_i,
\]
implying $\frac{1-\beta}{1+2\beta}h_i \leq l_i$
Now consider the hypothesis ball: $\mc B ^{H^-} _i$ set by the algorithm in the previous step (via Line \ref{alg:line:ziexp}). From \cref{lem:exp-factor}, we have $\mc B_i \subseteq \mc B ^{H^-} _i$.  Let $h_i ^ - = \radius\left(\mc B ^{H^-} _i\right)$. Since the algorithm shrinks the hypothesis ball by at least a factor of 2 during the zoom-in process, $h_i ^ -$ is at least $2h_i$. And therefore $2h_i \geq l_i$. From \cref{fig:almst-trckd} we also see that $s_i \leq h_i + l_i$. Therefore,
\begin{align*}
    \Phi_i = \log \left(\frac{\max(s_i,l_i,h_i)}{\sqrt{l_i h_i}}\right) &\leq \log \left(\frac{l_i + h_i}{\sqrt{l_i h_i}}\right)\\
           & \leq \log \left(\frac{2h_i + h_i}{\sqrt{\frac{1-\beta}{1+2\beta}h_i \cdot h_i}}\right)
             \leq \phi_0,\quad \text{for } \phi_0 = \log \left(3 \cdot \sqrt{\frac{1+2\beta}{1-\beta}}\right)
\end{align*}
\dave{It is a bit unclear what this refers to. The values of $l_i$ and $h_i$ during the current step or during the previous zoom-in step.}\adi{Added a figure, and rewrote the proof. Hopefully this is clearer}
\end{proof}

\phantomsection\label{prf:thm:lower-bound}
\begin{retheorem}{\ref{thm:lower-bound}}[Lower Bound for the Distance]
For any algorithm $\mc A$, there exists a starting configuration $Q \spr 0$ and an evolver (with knowledge of $\mc A$) in the local-motion model such that, for any positive integer $t_0$, there exists $t \geq t_0$, such that $\mc D \spr t = \mc D (\mc P \spr t, \mc H \spr t) \in \Omega(n)$.
\end{retheorem}
\begin{proof}
We construct a point set on the real line ($\RE$), and we will mention at the end how to adapt the proof to any dimension $d$. Given the small constant $\beta$, the local feature constant, we define $Q \spr 0$ to be the following set of $n = 2 m$ points in $\RE$.
\[
    Q \spr 0 
        ~ = ~ \bigcup_{i\in [m]} \left\{a_i\spr 0 = 100 i\right\} \cup \bigcup_{i\in [m]} \left\{b_i \spr 0 = 100 i + 1\right\}.
\]
Let $\mc D_{a,i} = \mc D_{2(k-1)+1},\; k \in [m]$, the distance corresponding to the first point in the tuple. We similarly define $P_{a,i}$ and $H_{a,i}$. Let $N_i$ denote the current distance to $a_i$'s nearest neighbor. The local feature interval for $a_i$, denoted $I_i$, is $100 i + \beta N_i [ -1, 1]$. Recall that $P_{a,i}$ is the uniform probability over $I_i$. Therefore, $P_{a,i}(\vctr x) = 1/|\mc I_i|$, where $|\mc I_i| = 1/\beta N_i$ is the diameter of $\mc I_i$. Now
\begin{flalign*}
    \mc D\spr {t_0} _{a,i} 
        & ~ = ~ \int_{\vctr x \in \mc I_i} P\spr {t_0}_{a,i}(\vctr x) \log \frac{P\spr {t_0}_{a,i}(\vctr x)}{H\spr {t_0}_{a,i}(\vctr x)} d\vctr x \\
        & ~ = ~ \int_{\vctr x \in \mc I_i} P\spr {t_0}_{a,i}(\vctr x) \log P\spr {t_0}_{a,i}(\vctr x) d\vctr x ~-~ \int_{\vctr x \in \mc I_i} P\spr {t_0}_{a,i}(\vctr x) \log H\spr {t_0}_{a,i}(\vctr x) d\vctr x.
\end{flalign*}
By the concavity of the log function and Jensen's inequality, we have
\begin{flalign}
    \mc D\spr {t_0} _{a,i}
        & ~ \geq ~ -\log|\mc I_i|\int _{\vctr x \in \mc I_i}P\spr {t_0} _{a,i}(\vctr x) d\vctr x ~-~ \log \int _{\vctr x \in \mc I_i} H\spr {t_0}_{a,i}(\vctr x)P\spr {t_0} _{a,i}(\vctr x) d\vctr x \nonumber\\
        & ~ = ~ -\log|\mc I_i| - \log \frac{ \int _{\vctr x \in \mc I_i} H\spr {t_0}_{a,i}(\vctr x) d\vctr x}{|\mc I_i|}
          ~ = ~ - \log \int_{\vctr x \in \mc I_i} H\spr {t_0}_{a,i}(\vctr x) d\vctr x. \label{Eqn:Cross-Entropy}
\end{flalign}

We may assume that $\mc D \spr {t_0} \in o(n)$ (for otherwise we can set $t = {t_0}$ and are done). This implies there are only $o(n)$ many $a_i$'s such that $\mc D\spr {t_0} _{a,i} \in \Omega(1)$. 
Call the remaining $n/2 - o(n)$, $a_i$'s the proximal set, denoted $T_a$.

We are now ready to describe the evolver's actions. It chooses to stay dormant until time ${t_0}$. For a large enough constant $M$, consider the time interval $[{t_0},{t_0}+n/M]$. The algorithm $\mc A$, running at a speedup factor of $\sigma$, can modify at most $\sigma n/M$ hypotheses in the time interval. The evolver chooses not to move any of those points. Therefore $\mc A$ can only reduce $\mc D \spr {t_0}$ by $o(n)$ over this time interval. Now there are at least $n/2 - \sigma n/M -o(n)$ members in the proximal set $T_a$, whose hypotheses were not altered by $\mc A$. Call this set the stable set, denoted $S_a$.
    
For $\kappa = \lceil\log(2/(1+\alpha))\rceil$, the evolver selects some $n/(\kappa M)$ members from $S_a$. Call that set $S'_a$. To be specific, for all $a_i \in S'_a$, the evolver chooses to move $a_i$ away from $b_i$ some $\kappa$ times, by a distance of exactly $\alpha N_{a,i}$, where $N_{a,i}$ is the distance of $a _i$ from its current nearest neighbor. (Note that the nearest neighbor will remain $b_i$ throughout these operations.) Given this value of $\kappa$, after the conclusion of these operations, the distance between $a_i$ and $b_i$ is at least 2. For $\beta < 1/3$, the local feature of $a_i$ changes to an interval $\mc I'_i$, where $\mc I'_i \cap \mc I_i = \emptyset$. Now, for $t={t_0}+n/M$, using a similar analysis as Eq.~\eqref{Eqn:Cross-Entropy}, we have 
\begin{equation}
    \mc D\spr t _{a,i} 
        ~ \geq ~ - \log \int_{\vctr x \in \mc I'_i} H\spr t_{a,i}(\vctr x) d\vctr x. \label{Eqn:Change-Entropy}
\end{equation}
Thus, for all $a_i \in S'_a$, we have $H\spr t _{a,i}= H\spr {t_0} _{a,i}$ and $D\spr {t_0} _{a,i} = o(1)$. Therefore, 
\[
    \int_{\vctr x \in \mc I_i} H\spr {t_0}_{a,i}(\vctr x) d\vctr x 
        ~ \geq ~ e^{-o(1)}.
\]
If $\mc D\spr t _{a,i} \in o(1)$ as well, then similarly from Eq.~\eqref{Eqn:Change-Entropy}, we have 
\[
    \int_{\vctr x \in \mc I'_i} H\spr t_{a,i}(\vctr x) d\vctr x 
        ~ \geq ~ e^{-o(1)}.
\]
But, this yields a contradiction since $\mc I'_i \cap \mc I_i = \emptyset$ and $H\spr t _{a,i}= H\spr {t_0} _{a,i}$ implies that
\[
    \int_{\vctr x \in \mc I'_i \cup \mc I_i} H\spr t_{a,i}(\vctr x) d\vctr x 
        ~ \geq ~ 2e^{-o(1)} 
        ~ > ~ 1.
\]
Therefore, for $a_i \in S'_a$, $\mc D\spr t _{a,i} \in \Omega(1)$, and since $|S'_a| = \Omega(n)$, we have $\mc D \spr t \in \Omega(n)$.

For a general dimension $d$, we define $Q$, in the same way except all the points lie on a single axis. The evolver also moves these points along that particular axis. The rest of the analysis involves integration over a region of space rather an interval, but is straightforward and carries over to $\RE^d$.
\end{proof}

\phantomsection\label{prf:lem:EvoStep}
\begin{relemma}{\ref{lem:EvoStep}}[Evolver's contribution to $\Phi$]
    Every step of the evolver increases the potential function $\Phi$ by a constant.
\end{relemma}
\begin{proof}
    Let's say the evolver chooses to move $\vctr q_i$ at time $0$. Now $l \spr 0 _i/\beta$ is the nearest neighbor distance of $\vctr q_i$ at time $0$. Therefore $\vctr q_i$ moves by at most $(\alpha/\beta) l\spr 0_i$ distance, implying $s_i \spr 1 \leq s_i \spr 0 + (\alpha/\beta) l_i \spr 0 \leq \left(1+ \alpha/\beta\right)\max \left(s_i \spr 0,l_i \spr 0,h_i \spr 0\right)$. Similarly $(1-\alpha)l_i \spr 0 \leq l_i \spr 1 \leq (1+\alpha)l_i \spr 0$. Therefore,

\begin{flalign}
\begin{split}
    \Phi_i \spr 1 - \Phi_i \spr 0 & ~\leq~ \log \frac{\max \left(s_i \spr 1,l_i \spr 1,h_i \spr 1\right)}{\sqrt{l_i \spr 1 h_i \spr 1}} 
        - \log \frac{\max \left(s_i \spr 0,l_i \spr 0,h_i \spr 0\right)}{\sqrt{l_i \spr 0 h_i \spr 0}}\\
        & ~\leq~ \log \frac{\max \left(s_i \spr 0 + (\alpha/\beta) l_i \spr 0,(1+\alpha)l_i \spr 0,h_i \spr 0\right)}{\sqrt{(1-\alpha)l_i \spr 0 h_i \spr 0}} 
        ~ - ~ \log \frac{\max \left(s_i \spr 0,l_i \spr 0,h_i \spr 0\right)}{\sqrt{l_i \spr 0 h_i \spr 0}}\\
        & ~\leq~ \log \left(\frac{1+\alpha/\beta}{\sqrt{1-\alpha}} \frac{\max \left(s_i \spr 0,l_i \spr 0,h_i \spr 0\right)}{\sqrt{l_i \spr 0 h_i \spr 0}}\right)
        ~ - ~\log \frac{\max \left(s_i \spr 0,l_i \spr 0,h_i \spr 0\right)}{\sqrt{l_i \spr 0 h_i \spr 0}}\\
        & ~ = ~ \log \frac{1+\alpha/\beta}{\sqrt{1-\alpha}} \in O(1) \label{Eqn:EvoPotSelf}
\end{split}
\end{flalign}

Let $\mc N_i$ be the set of indices of points whose nearest neighbor at time 0 was $\vctr q_i\spr 0$, or whose nearest neighbor at time 1 was $\vctr q_i \spr 1$. Only points with indices in this set have their potential function changed. If $j \in \mc N_i$, $s_j \spr 0$, and $h_j \spr 0$ do not change. Since $\vctr q_i$  moves by at most $(\alpha/\beta) l\spr 0_i$ distance, the nearest neighbor distance $N_j = l_j/\beta$ changes by at most $(\alpha/\beta) l\spr 0_i$ as well. Therefore,
\begin{gather}
        -(\alpha/\beta)l_i \spr 0 ~ \leq ~ N_j \spr 1 - N_j \spr 0 ~ \leq ~ (\alpha/\beta)l_i \spr 0 \notag\\
        \implies -\alpha l_i \spr 0 ~ \leq ~ l_j \spr 1 - l_j \spr 0 ~ \leq ~ \alpha l_i \spr 0 \label{Eqn:LDiff1}
        \intertext{Since $(1-\alpha)l_i \spr 0 ~ \leq ~ l_i \spr 1$, we also have}
        -\frac{\alpha}{1-\alpha}l_i \spr 1 ~\leq ~ l_j \spr 1 - l_j \spr 0 ~ \leq ~ \frac{\alpha}{1-\alpha}l_i \spr 1 \label{Eqn:Ldiff2}
\end{gather}

If $\vctr q_i \spr 0$ was the nearest neighbor of $\vctr q_j \spr 0$, then $N_j \spr 0 \leq N_i \spr 0$, and hence $l_j \spr 0 \leq l_i \spr 0$. Using the last fact, and dividing Eq.~\eqref{Eqn:LDiff1} by $l_j \spr 0$, we have $(1-\alpha) \leq l_j \spr 1 / l_j \spr 0 \leq (1+\alpha)$. Similarly, If $\vctr q_i \spr 1$ was the nearest neighbor of $\vctr q_j \spr 1$, then using Eq.~\eqref{Eqn:Ldiff2} we have $(1-\alpha/(1-\alpha)) \leq l_j \spr 1 / l_j \spr 0 \leq (1+\alpha/(1-\alpha))$. Therefore $ l_j \spr 1 / l_j \spr 0 \in \Theta(1)$ for $j \in \mc N_i$. And hence,

\begin{flalign}
\begin{split}
    \Phi_j \spr 1 - \Phi_j \spr 0 & ~ = ~ \log \frac{\max \left(s_i \spr 1,l_i \spr 1,h_i \spr 1\right)}{\sqrt{l_i \spr 1 h_i \spr 1}} 
        ~ - ~ \log \frac{\max \left(s_i \spr 0,l_i \spr 0,h_i \spr 0\right)}{\sqrt{l_i \spr 0 h_i \spr 0}}\\
        & ~ = ~ \log \frac{\max \left(s_i \spr 0,l_i \spr 1,h_i \spr 0\right)}{\max \left(s_i \spr 0,l_i \spr 0,h_i \spr 0\right)} 
        ~ - ~ \log \frac{\sqrt{l_i \spr 1 h_i \spr 0}}{\sqrt{l_i \spr 0 h_i \spr 0}}\\
        & ~ \leq ~ \left| \log \frac{l_i \spr 1}{l_i \spr 0} \right| ~ + ~ \left|\log \frac{l_i \spr 0}{l_i \spr 1}\right| ~ \in ~ O(1), \quad \forall ~ j \in \mc N_i\label{Eqn:EvoPotOth}
\end{split}
\end{flalign}

Finally, we show that $|\mc N_i| \in O(1)$. To that effect we solve a general problem: What is the maximum number of points in $Q$, whose nearest neighbor is $\vctr q_1$? Without loss of generality, let $\vctr q_2$ be the nearest neighbor of $\vctr q_1$, such that $\|\vctr q_2 - \vctr q_1\| = 1$. Let $\vctr q_3$'s nearest neighbor be $\vctr q_1$. Consider the line segment $\overline{\vctr q_1 \vctr q_3}$, and its intersection with $\mc B (\vctr q_1, 1)$ the ball centered at $\vctr q_1$, and passing through $\vctr q_2$. Call the intersection point $\vctr q_3'$. For any general point $\vctr q_x \in Q$, we similarly define $\vctr q_x'$. Note $\vctr q_2' = \vctr q_2$. By triangle inequality, we have $\|\vctr q_3' - \vctr q_2'\| \geq \|\vctr q_3 - \vctr q_2\| - \|\vctr q_3 - \vctr q_3'\| = \|\vctr q_3 - \vctr q_2\| - \|\vctr q_3 - \vctr q_1\| + \|\vctr q_1 - \vctr q_3'\|$. Since $\vctr q_1$ is the nearest neighbor of $\vctr q_3$, we have $\|\vctr q_3 - \vctr q_2\| \geq \|\vctr q_3 - \vctr q_1\|$. Therefore $\|\vctr q_3' - \vctr q_2'\| \geq \|\vctr q_1 - \vctr q_3'\| = 1$. Therefore, for any $x,y\in [n], x\neq y$, such that $\vctr q_1$ is the nearest neighbor of $\vctr q_x$ and $\vctr q_y$, we have $\|\vctr q_x' - \vctr q_y'\| \geq 1$. Let $\mc M_1 = \{\vctr q_x' \mid \vctr q_1 \text{is the nearest neighbor of } \vctr q_x\}$. Draw a ball of radius 1/2 around $\vctr q_1$, and every member of $\mc M_1$. Each of these balls are non-overlapping. However all of them are contained within a ball around $\vctr q_1$ of radius $3/2$. Therefore $|\mc M_1| < (3/2)^d/(1/2)^d = 3^d \in O(1)$. Observing that $\mc N_i \leq 2\mc M_1$, and using Eq.~\eqref{Eqn:EvoPotSelf}, and \ref{Eqn:EvoPotOth} we have the result.
\end{proof}

\phantomsection\label{prf:lem:alg-steps}
\begin{relemma}{\ref{lem:alg-steps}}[Algorithm's contribution towards $\Phi$]
    For any iteration $z$, let $\eps \itr z$ denote the number steps executed by the evolver. Then \ref{Alg:TrackByZoom} increases $\Phi$ in $O\left(n+\eps \itr z\right)$ steps, each time by $O(1)$. In each of the remaining steps of the $z$th iteration it decreases $\Phi$ by $\Theta(1)$ in an amortized sense.
\end{relemma}
\begin{proof}
    For a particular index $i$, we first concentrate on the zoom-out routine of our \cref{Alg:TrackByZoom} (Line \ref{alg:line:zoomout}). Assume that the evolver leaves $q_i$ untouched for the time being. We later show how to handle the case when evolver moves $q_i$, while the algorithm is processing the index $i$. The algorithm increases $h_i$ by a constant factor in Line \ref{alg:line:kzo}. Therefore $\Phi = \log (\max (s_i,l_i,h_i)/\sqrt{l_i h_i})$ decreases by a constant amount whenever $\max (s_i,l_i,h_i) \neq h_i$. Now suppose $\max (s_i,l_i,h_i) = h_i$ at time $t$, for the first time during the zoom-out routine. That implies a 2-expansion of $\mc B^H_i$ in line \ref{alg:line:kzo} contains the entire $\mc B_i$. Therefore, a further $1/\beta$-expansion definitely contains another $X_j$, $j\neq i$, satisfying the condition in line \ref{alg:line:exzoc}. We conclude that once $\max (s_i,l_i,h_i) = h_i$, the algorithm only increases $\Phi$ by a constant amount before moving on to the zoom-in routine of the algorithm (line \ref{alg:line:zoomin}).

    Now, we look at the zoom-in routine of the algorithm (line \ref{alg:line:zoomin}). 
    The algorithm only zooms in as long as the condition in line \ref{alg:line:kzic} is satisfied. 
    Since we expand $\mc B_i^H$ with an expansion factor $\frac{3}{1-2\beta}$ in line \ref{alg:line:ziexp}, using \cref{lem:exp-factor} we conclude that $\mc B_i \subset \mc B_i^H$ in every step during the zoom-in routine. 
    That implies $\max (s_i,l_i,h_i) = h_i$ throughout the zoom-in process, further implying $\Phi_i = \log\sqrt{h_i/l_i}$. 
    In line \ref{alg:line:zinested} we subdivide the hypothesis ball into $|\Psi| = \left(2\lceil \frac{3}{1-2\beta} \rceil\sqrt{d}\right)^d \in O(1)$ balls and make constant number of queries to the oracle. 
    Therefore $h_i$ reduces by a factor of $\frac{3}{1-2\beta}/\left(2\lceil \frac{3}{1-2\beta} \rceil\right) \geq 2$. 
    And hence $\Phi$ reduces by at least a constant amount. We amortize the reduction over the $|\Psi|$ calls to the Oracle, to denote a constant reduction in $\Phi$ in every zoom-in step. \adi{amortized the reduction here} 
    
    Now if the evolver moved $q_i$ while the algorithm was processing the index $i$, there is a possibility that $X_i$ might move outside of $\mc B^H_i$ during the zoom-in routine. We take care of this condition in line \ref{alg:line:exevzic} to initiate the zoom-out process. Since, the algorithm possibly increases $\Phi$ by a constant amount when switching from zoom-out to zoom-in, we conclude: For every evolver's step, our algorithm might increase $\Phi$ by at most a constant amount as well.
\end{proof}

\phantomsection\label{prf:lem:iter-time-pot}
\begin{relemma}{\ref{lem:iter-time-pot}}[Iteration time proportional to Potential]
    There exists a constant speedup factor $\sigma$ for \ref{Alg:TrackByZoom} such that $\Delta \itr z \in \Theta_\sigma\left(n+\Phi \itr z\right)$.
\end{relemma}
\begin{proof}
    By \cref{lem:EvoStep}, the evolver can only increase the potential by $\Phi_\eps \itr z  = O(\Delta \itr z)$ during the iteration. 
    By \cref{lem:alg-steps} we observe that the algorithm increases $\Phi$ by some $O(n+ \Delta \itr z)$ amount. 
    Since by \cref{lem:almst-trckd}, the algorithm reduces each $\Phi_i$ to at most $\phi_0$, it reduces the overall potential by at most $\Phi \itr z - n \phi_0 + \Phi_\eps \itr z + O(n+\Delta \itr z)$ $= O(n+\Phi \itr z + \Delta \itr z)$. 
    For a speedup factor $\sigma$, the algorithm \ref{Alg:TrackByZoom} spends $O(n+\Phi \itr z + \Delta \itr z) /\sigma$ time to do so, implying $\Delta \itr z \leq O(n+\Phi \itr z + \Delta \itr z)/\sigma$, further implying $\Delta \itr z \in O_\sigma(n+\Phi \itr z)$, for a large enough constant $\sigma$.

    To see why $\Delta \itr z \in \Omega(n+\Phi \itr z)$, we follow a similar logic except observe that the evolver can decrease the potential by at most $\Phi_\eps \itr z  = O(\Delta \itr z)$ as well. 
    By \cref{lem:almst-trckd} then, the algorithm reduces the overall potential $\Phi$ by at least $\Phi \itr z - n \phi_0 - \Phi_\eps \itr z = \Omega(n+\Phi \itr z - \Delta \itr z)$, which takes it at least $\Omega(n+\Phi \itr z - \Delta \itr z)/\sigma$ time implying $\Delta \itr z \geq \Omega(n+\Phi \itr z - \Delta \itr z)/\sigma$, resulting in $\Delta \itr z \in \Omega_\sigma(n+\Phi \itr z)$.
\end{proof}


\end{document}